\documentclass{amsart}
\usepackage[margin=2.8cm]{geometry}
\usepackage{amssymb}
\usepackage{graphicx} 
\usepackage{mathrsfs}
\usepackage{enumerate}
\usepackage{xspace}
\usepackage{color}
\usepackage{enumerate}
\usepackage{enumitem}
\usepackage{amsthm}
\usepackage{dsfont}
\usepackage{bbm}
\usepackage[colorlinks=true, linkcolor = blue, citecolor = blue]{hyperref}
\usepackage[numbers]{natbib}
\usepackage[backgroundcolor=white,bordercolor=red]{todonotes}
\DeclareMathAlphabet{\mathpzc}{OT1}{pzc}{m}{it}
\usepackage[nameinlink]{cleveref}
\numberwithin{equation}{section}
\begin{document}

\theoremstyle{plain}

\newtheorem{theorem}{Theorem}[section]
\newtheorem{lemma}[theorem]{Lemma}
\newtheorem{example}[theorem]{Example}
\newtheorem{proposition}[theorem]{Proposition}
\newtheorem{corollary}[theorem]{Corollary}
\newtheorem{definition}[theorem]{Definition}
\newtheorem{Ass}[theorem]{Assumption}
\newtheorem{condition}[theorem]{Condition}
\theoremstyle{definition}
\newtheorem{remark}[theorem]{Remark}
\newtheorem{SA}[theorem]{Standing Assumption}

\newcommand{\of}{[\hspace{-0.06cm}[}
\newcommand{\gs}{]\hspace{-0.06cm}]}

\newcommand\llambda{{\mathchoice
		{\lambda\mkern-4.5mu{\raisebox{.4ex}{\scriptsize$\backslash$}}}
		{\lambda\mkern-4.83mu{\raisebox{.4ex}{\scriptsize$\backslash$}}}
		{\lambda\mkern-4.5mu{\raisebox{.2ex}{\footnotesize$\scriptscriptstyle\backslash$}}}
		{\lambda\mkern-5.0mu{\raisebox{.2ex}{\tiny$\scriptscriptstyle\backslash$}}}}}

\newcommand{\1}{\mathds{1}}

\newcommand{\F}{\mathbf{F}}
\newcommand{\G}{\mathbf{G}}
\newcommand{\B}{\mathbf{B}}
\newcommand{\cF}{\mathcal{F}}
\newcommand{\cG}{\mathcal{G}}
\newcommand{\cN}{\mathcal{N}}

\newcommand{\M}{\mathcal{M}}

\newcommand{\la}{\langle}
\newcommand{\ra}{\rangle}

\newcommand{\lle}{\langle\hspace{-0.085cm}\langle}
\newcommand{\rre}{\rangle\hspace{-0.085cm}\rangle}
\newcommand{\blle}{\Big\langle\hspace{-0.155cm}\Big\langle}
\newcommand{\brre}{\Big\rangle\hspace{-0.155cm}\Big\rangle}

\newcommand{\X}{\mathsf{X}}

\newcommand{\tr}{\operatorname{tr}}
\newcommand{\N}{{\mathbb{N}}}
\newcommand{\cadlag}{c\`adl\`ag }
\newcommand{\on}{\operatorname}
\newcommand{\oP}{\overline{P}}
\newcommand{\oO}{\mathcal{O}}
\newcommand{\D}{D(\mathbb{R}_+; \mathbb{R})}
\newcommand{\bS}{\mathbb{S}}

\renewcommand{\epsilon}{\varepsilon}

\newcommand{\fPs}{\mathfrak{P}_{\textup{sem}}}
\newcommand{\fPas}{\mathfrak{P}^{\textup{ac}}_{\textup{sem}}}
\newcommand{\rrarrow}{\twoheadrightarrow}
\newcommand{\cC}{\mathcal{C}}
\newcommand{\cH}{\mathcal{H}}
\newcommand{\cD}{\mathcal{D}}
\newcommand{\cE}{\mathcal{E}}
\newcommand{\cP}{\mathcal{P}}
\newcommand{\cR}{\mathcal{R}}
\newcommand{\cQ}{\mathcal{Q}}
\newcommand{\cM}{\mathcal{M}}
\newcommand{\cU}{\mathcal{U}}
\newcommand{\bth}{\overset{\leftarrow}\theta}
\renewcommand{\th}{\theta}
\newcommand{\cA}{\mathcal{A}}
\newcommand{\fP}{\mathfrak{P}}
\newcommand{\fM}{\mathfrak{M}}
\newcommand{\bx}{\mathsf{x}}

\newcommand{\bR}{\mathbb{R}}
\newcommand{\bN}{\mathbb{N}}
\newcommand{\nnabla}{\nabla}
\newcommand{\f}{\mathfrak{f}}
\newcommand{\g}{\mathfrak{g}}
\newcommand{\oconv}{\overline{\operatorname{conv}}\hspace{0.1cm}}
\newcommand{\usa}{\on{usa}}
\newcommand{\usc}{\on{USC}}
\newcommand{\limmed}{\on{\hspace{0.025cm}lim \hspace{0.05cm} med \hspace{0.05cm}}}
\newcommand{\esssup}{\on{\hspace{0.025cm}ess\hspace{0.05cm} sup \hspace{0.05cm}}}

\newcommand{\C}{\mathsf{C}}
\newcommand{\K}{\mathsf{K}}
\newcommand{\ou}{\overline{u}}
\newcommand{\ua}{\underline{a}}
\newcommand{\uu}{\underline{u}}

\newcommand{\NFLVR}{\(\on{NFLVR}(\cP)\)\hspace{0.05cm}}

\renewcommand{\emptyset}{\varnothing}

\allowdisplaybreaks

\makeatletter
\newcommand{\mylabel}[2]{#2\def\@currentlabel{#2}\label{#1}}
\makeatother

\makeatletter
\@namedef{subjclassname@2020}{%
	\textup{2020} Mathematics Subject Classification}
\makeatother

 \title[]{Robust utility maximization with nonlinear \\continuous semimartingales}
\author[D. Criens]{David Criens}
\author[L. Niemann]{Lars Niemann}
\address{Albert-Ludwigs University of Freiburg, Ernst-Zermelo-Str. 1, 79104 Freiburg, Germany}
\email{david.criens@stochastik.uni-freiburg.de}
\email{lars.niemann@stochastik.uni-freiburg.de}



\thanks{We are grateful to an associate editor and two anonymous referees for many valuable comments and suggestions that helped us to improve the paper.} 
\thanks{DC acknowledges financial support from the DFG project SCHM 2160/15-1 and LN acknowledges financial support from the DFG project SCHM 2160/13-1.}
\date{\today}

\maketitle

\begin{abstract}
In this paper we study a robust utility maximization problem in continuous time under model uncertainty. 
The model uncertainty is governed by a continuous semimartingale with uncertain local characteristics. Here, the differential characteristics are prescribed by a set-valued function that depends on time and path.
We show that the robust utility maximization problem is in duality with a conjugate problem, and we study the existence of optimal portfolios for logarithmic, exponential and power utilities.
\end{abstract}

\section{Introduction}
\subsection{The purpose of this article}
An important problem for a portfolio manager is to maximize the expected utility of his terminal wealth. For complete markets, this problem can be solved by the martingale method developed in \cite{CH89,CH91,KLS87,pliska86}. 
The case of incomplete markets is considerably more difficult. By now, the classical approach to compute the maximized utility (called \emph{value function}) is to use a \emph{duality argument} which was formalized in an abstract manner in the seminal paper \cite{kramkov}, see also \cite{HP91,KLSX91} for other pioneering work. 
The key idea is to pass to a \emph{dual optimization problem}, which is typically of reduced complexity, and to recover, via a bidual relation, the value function of the original problem as the conjugate of the value function corresponding to the dual problem. 

In this paper we are interested in a robust framework, where, instead of a single financial model, a whole family of models is taken into consideration. Financially speaking, we think of a portfolio manager who is uncertain about the real-world measure, but who thinks that it belongs to a certain set of probabilities.

Recently, an abstract duality theory for possibly nondominated robust frameworks was developed in \cite{kupper}, see also \cite{denis} for another approach to a robust duality theory under drift and volatility uncertainty for bounded utilities. 
Compared to the classical case treated in \cite{kramkov}, the theory from \cite{kupper} relies on a measure-independent dual pairing which requires a suitable topological structure. A natural choice for an underlying space is the Wiener space of continuous functions, which can be seen as a \emph{canonical framework}.

In their fundamental work \cite{kramkov}, the authors use their abstract theory to establish duality theorems for general semimartingale market models. When it comes to robust duality theorems in nondominated settings, it seems that only the L\'evy type setting with deterministic uncertainty sets has been studied in detail, see \cite{kupper,denis}. The purpose of this paper is to reduce the gap between the robust and non-robust case in terms of a robust duality theory for nondominated canonical continuous semimartingale markets with time and path-dependent uncertainty sets. 
To this end, we rely on the ideas and abstract results of \cite{kupper} that allow us to derive robust duality theorems for a larger class of stochastic models. 
Therefore, we place ourselves in the continuous path-setting of \cite{kupper}. In the following, we explain our setting in more detail.

\subsection{The setting}
Consider the robust utility maximization problem given by
\begin{equation} \label{eq: intro u}
    u(x) := \sup_{g \in \cC(x)} \inf_{P \in \cP} E^P\big[ U(g) \big],
\end{equation}
where \(U \colon (0, \infty) \to \bR\) is a utility function,
\(\cP\) is a set of (possibly nondominated) probability measures on the Wiener space \(\Omega := C([0,T]; \bR^d)\), with finite time horizon \(T > 0\), and
\(\cC(x) := x\hspace{0.05cm} \cC\) is the set of claims that can be \(\cP\)-quasi surely superreplicated  with initial capital \(x > 0\), i.e.,
\begin{align*}
 \cC  := \Big \{ g \colon \Omega \to [0, \infty] \colon  \text{ \(g\) is universally measurable and }  \exists H \in \cH^\cP \text{ with } 1 +\int_0^T H_s d X_s \geq g \ \cP\text{-q.s.} \Big\}.
\end{align*}
The model uncertainty in this framework is introduced through a set \(\cP\), which consists of semimartingale laws on the Wiener space. We parameterize \(\cP\) via a compact parameter space \(F\) and drift and volatility coefficients \(b \colon F \times [0, T] \times \Omega \to \bR^d\) and \(a \colon F \times [0, T] \times \Omega \to \mathbb{S}^d_+\) such that
\[  
\cP = \big\{ P \in \fPas \colon P \circ X_{0}^{-1} = \delta_{x_0},\ (\llambda \otimes P)\text{-a.e. } (dB^{P} /d\llambda, dC^{P}/d\llambda) \in \Theta   \big\},
\]
where \(\fPas\) denotes the set of semimartingale laws with absolutely continuous characteristics, \(X\) is the coordinate process, \(x_0 \in \bR^d\) is the initial value, \((B^P, C^P)\) are the \(P\)-characteristics of \(X\), and
\[
\Theta (t, \omega) := \big\{(b (f, t, \omega), a (f, t, \omega)) \colon f \in F \big\}, \quad (t, \omega) \in [0,T] \times \Omega.
\]
In words, the set \(\cP\) of feasible real-world measures consists of all continuous semimartingale models whose coefficients take uncertain values in the fully path-dependent set \(\Theta\). 
This time and path dependence constitutes the main novelty of our paper, extending the results of \cite{kupper} where \(\Theta (t, \omega) \equiv \Theta\) is independent of time \(t\) and path \(\omega\).

We prove our main convex duality results under the assumption that \(b\) and \(a\) are continuous and of linear growth, and that \(\Theta\) is convex-valued. Further, we will introduce a robust market price of risk, which seems to be a novel object. In case we deal with unbounded utility functions, we additionally assume either a certain uniform boundedness condition or that the volatility coefficient \(a\) is uniformly bounded and elliptic. In general, however, we do not impose any ellipticity assumption and thence allow the portfolio manager to take incomplete markets into consideration. It seems to us that this feature is new for robust semimartingale frameworks.

Due to the high amount of flexibility of our framework, we are able to cover many prominent stochastic models. This includes the case from \cite{kupper} where \(\Theta (t, \omega) \equiv \Theta\) is independent of time \(t\) and path \(\omega\), which corresponds to the generalized \(G\)-Brownian motion as introduced in \cite{peng2010}, cf. also \cite{neufeld2017nonlinear} for a nonlinear L\'evy setting {\em with jumps}. 
Additionally, we are able to capture a Markovian framework of nonlinear diffusions 
 where \(\Theta (t, \omega) \equiv \Theta (\omega (t))\) depends on \((t, \omega)\) only through the value \(\omega (t)\). Such models have been investigated, for instance, in \cite{CN22b, hol16}.
 Furthermore, our setting can also be used to model path-dependent dynamics such as stochastic delay equations under parameter uncertainty and the random \(G\)-expectation  as discussed in Section 4 from \cite{NVH}, cf. also \cite{nutz} for a related approach.

\subsection{Main contributions} 

Denote the set of absolutely continuous separating measures by
\[
    \cD  := \big\{ Q \in \mathfrak{P}_a(\cP) \colon E^Q [ g ] \leq 1 \text{ for all } g \in \cC \big\}, \quad \cD(y)  := y\hspace{0.05cm} \cD, \ \ y >0,
\]
where \(\mathfrak{P}_a(\cP) := \{ Q \in \mathfrak{P}(\Omega) \colon \exists P \in \cP \text{ with } Q \ll P \}\).
The robust dual problem is given by 
\begin{equation} \label{eq: intro v}
    v(y)  := \inf_{Q \in \cD(y)} \sup_{P \in \cP} E^P\Big[V\Big(\frac{dQ}{dP}\Big) \Big],
\end{equation}
where \(V\) denotes the conjugate of the utility function \(U\).
We focus on logarithmic, exponential and power utility, i.e., \(U\) is assumed to be one of the following
\[
U(x) = \log(x), \quad U(x) = -e^{-\lambda x} \text{ for } \lambda > 0, \quad U(x) = \frac{x^p}{p} \text{ for } p \in (-\infty,0) \cup (0,1). 
\]
For these utility functions we show that the functions \(u\) and \(v\) are conjugates, i.e., 
\begin{equation} \label{eq: intro conjugates}
    u(x) = \inf_{y > 0} \big[v(y) + xy \big], \ \ x > 0, \quad v(y) = \sup_{x > 0} \big[u(x) - xy \big], \ \ y > 0,
\end{equation}
which constitutes our main result. 
Additionally, we prove the existence of an optimal portfolio for a relaxed version of the optimization problem \eqref{eq: intro u}, which accounts for the obstacle that in the nondominated case one cannot rely on classical tools like Koml\'os' lemma for the approximation scheme of an optimal portfolio. 
In order to show \eqref{eq: intro conjugates}, we use the duality results developed in \cite{kupper} and adapt the strategy laid out in \cite[Section~3]{kupper} beyond the case of nonlinear continuous L{\'e}vy processes, i.e., where the set \(\Theta\) is independent of time and path.
	
To apply the main duality results from \cite{kupper}, we prove that \(\cP\) and \(\cD\) are convex and compact, and
that the sets \( \cC\) and \( \cD\) are in duality, i.e.,
\begin{equation} \label{eq: intro first duality}
	\big\{ Q \in \mathfrak{P}_a(\cP) \colon E^Q [g ] \leq 1 \text{ for all } g \in \cC \cap C_b  (\Omega; \mathbb{R})\big \} = \cD,   
\end{equation}
and 
\begin{equation} \label{eq: intro second duality}
	\big\{ g \in C_b^+(\Omega; \bR) \colon E^Q [g] \leq 1 \text{ for all } Q \in \cD\big\} 
	= \cC \cap C_b (\Omega; \bR).
\end{equation}
We emphasise that for these dualities we work with absolutely continuous measures, while equivalent measures were used in \cite[Section~3]{kupper}.
As we explain in the following, we establish \eqref{eq: intro second duality} under a robust no arbitrage condition, which appears to us very natural. The corresponding duality from \cite[Section~3.1]{kupper} is proved under a uniform ellipticity assumption that implies the robust no arbitrage condition for the duality \eqref{eq: intro second duality} and that all absolutely continuous martingale measures are already equivalent martingale measures.
We now comment in more detail on the proofs.

For the first duality \eqref{eq: intro first duality}, we show that \(\cD\) coincides with the robust analogue of the set of absolutely continuous local martingale measures, i.e., 
\[
\fM_a(\cP) := \big\{ Q \in \fP_a(\cP) \colon X \text{ is a local \(Q\)-\(\F\)-martingale} \big\}.
\]
The equality \( \fM_a(\cP) = \cD \) relies on a characterization of local martingale measures on the Wiener space, and it
resembles the fact that for continuous paths, the set of separating measures coincides with the set of local martingale measures. 

Regarding the  second duality \eqref{eq: intro second duality}, the equality \( \fM_a(\cP) = \cD \) further allows us to use the robust superhedging duality 
\begin{align} \label{eq: SHD intro}
\sup_{Q \in \fM_a(\cP)} E^Q \big[f \big] = \min\Big\{ x \in \bR \colon \exists H \in \cH^{\fM_a(\cP)} \text{ with } x + \int_0^T H_s d X_s \geq f ~ Q\text{-a.s. \(\forall\) } Q \in \fM_a(\cP) \Big\},
\end{align}
to show that there is a superhedging strategy for every function in the polar of \(\cD\). 
To prove the duality \eqref{eq: SHD intro} we heavily rely on ideas from \cite{nutz_superhedging} and establish stability properties of a time and path-dependent correspondence related to \(\fM_a (\cP)\). As already mentioned above,
we are able to establish the
duality \eqref{eq: intro second duality} without imposing ellipticity conditions as used in \cite{kupper,denis}.
Rather, we work under a robust \emph{no free lunch with vanishing risk} condition that ensures that the set \(\fM_a(\cP)\) is sufficiently rich. More precisely, we assume that for every \(P \in \cP\) there exists a measure \(Q \in \fM_a(\cP)\) with \(P \ll Q\). This assumption is a continuous time version of the robust no-arbitrage condition introduced in \cite{nutz_nondom} and reduces to the classical no free lunch with vanishing risk condition in case \(\cP\) is a singleton.

Next, we comment on the proofs for convexity and compactness of the sets \(\cP\) and \(\cD\).
For \(\cP\) we adapt a strategy from \cite{CN22b} from a one-dimensional nonlinear diffusion setting to our multidimensional path-dependent framework. 
To establish compactness of the set \(\cD\), we show that it coincides with 
\[
\cM := \big\{ Q \in \fPas \colon Q \circ X_0^{-1} = \delta_{x_0},\ (\llambda \otimes Q)\text{-a.e. } (dB^{Q} /d\llambda, dC^{Q}/d\llambda) \in \tilde{\Theta}   \big\},
\]
where
\[
\tilde{\Theta} (t, \omega) := \{0\}^d \times \big\{a (f,t, \omega) \colon f \in F \big\} \subset \mathbb{R}^d \times \mathbb{S}^d_{+}, \quad (t, \omega) \in [0,T] \times \Omega.
\]
Compactness of \(\cM\) then can be proved as for its companion \(\cP\). 
To derive the equality \(\cD = \cM\), we assume the existence of a \emph{robust market price of risk (MPR)}. That is, the existence of a Borel function \(\theta \colon F \times [0, T] \times \Omega \to \mathbb{R}^d\), subject to a modest boundedness assumption, such that \( b = a \theta\). The robust MPR allows for equivalent measure changes between the set of candidate measures \(\cP\) and the robust version of the set of martingale measures \(\cM\).
By means of an example, we show that mild boundedness assumptions on the MPR are necessary for the identity \(\mathcal{M} = \mathcal{D}\) to hold.

Evidently, by virtue of \eqref{eq: intro v}, to deal with unbounded utility functions, we require some integrability of the Radon--Nikodym derivative \(dQ/dP\) 
for \(Q \in \cD\) and \(P \in \cP\).
To this end, we establish finite polynomial moments for certain stochastic exponentials. We think that this result is of independent interest.
In order to achieve this, we give a boundedness condition on the MPR and a uniform ellipticity and boundedness condition on the volatility from which we require only one to hold.
The first allows us to incorporate incomplete market models, while the second gives additional freedom in the drift coefficient.

\subsection{Comments on related literature}
There is already a vast literature on the robust utility maximization problem \eqref{eq: intro u} with respect to nondominated probability measures. 
In the discrete-time setting, the robust utility maximisation problem has been considered for instance in \cite{bartl,BC18,COW19, nutz_utility}, see also the references therein.
Nonlinear L\'evy frameworks with constant \(\Theta\) were for instance considered in \cite{kupper, denis, lin,nutz_levy} and nonlinear time-inhomogeneous L\'evy settings with time-dependent \(\Theta\) were studied in \cite{liang,liang2,park}. Further, \cite{biagini} investigated a robust Merton problem with uncertain volatility of L\'evy type and volatility state dependent drift.
Compared to these papers, we allow for uncertain drift and volatility with fully path dependent coefficients. In particular, our work includes diffusion models with uncertain parameters, such as real-valued nonlinear affine models as studied in \cite{fadina2019affine}, as well as stochastic volatility models with uncertain volatility processes.

The literature on the robust dual problem in continuous time is less extensive. Indeed, we are only aware of the papers \cite{kupper,denis}, where the problem is investigated from a theoretical perspective, and the only concrete examples we know of are the continuous L\'evy frameworks with uniformly bounded and elliptic coefficients as discussed in these papers. 

In the remainder of this subsection we comment on the differences between our proofs and those from~\cite{kupper} for the L\'evy setting. 
To establish \eqref{eq: intro conjugates} -- \eqref{eq: intro second duality}, we follow the ideas used in \cite[Section 3]{kupper} for the L\'evy case. 
We do however, replace the uniform ellipticity assumption of \cite[Section 3]{kupper} with a robust notion of \emph{no free lunch with vanishing risk}. This is in the spirit of the seminal work \cite{kramkov}. Further, notice that we work with the set 
\(\fM_a(\cP)\) of absolutely continuous local martingale measures as in \cite{nutz_nondom}. This relaxes our condition of no free lunch with vanishing risk, compared to its counterpart  formulated with equivalent local martingale measures.
To take our unbounded, non-elliptic and path-dependent framework into account, we have to develop new results on the equivalence, convexity and compactness of the sets \(\mathcal{P}\) and \(\mathcal{M}\). The difficulty of extending results from a L\'evy setting to its more general path-dependent counterparts has already been acknowledged in the literature (see, e.g., \cite{hol16,K21}).
Another novelty in our treatment is the concept of a robust MPR that is crucial to capture incomplete market situations that are excluded in \cite[Section 3]{kupper}. 
Related to the MPR, we need to investigate the martingale property and establish polynomial moment estimates for certain stochastic exponentials. To achieve this, the mere existence of a market price of risk is not sufficient and we have to impose either an additional boundedness assumption on the MPR, or restrict ourselves to a uniformly elliptic setting. Both conditions are satisfied for the L\'evy framework studied in \cite[Section~3]{kupper}. Finally, let us emphasize again that this is the only point were we use an ellipticity condition.

\subsection{Structure of the article}
In Section \ref{sec: setting} we lay out our setting, in Section \ref{sec: superhedging duality} we provide the superhedging duality for nonlinear continuous semimartingales
and in Section \ref{sec: separating duality} we discuss the duality relation between  \(\cC\) and \(\cD\). 
After stating the separating dualities and the appropriate notion of \emph{no-arbitrage}, we give parameterized conditions that ensure compactness and convexity of \(\cP\) and \(\cM\), respectively. The final part of Section \ref{sec: separating duality} is devoted to the robust market price of risk and to the equality \(\cD = \cM\).
In Section \ref{sec: duality utility} we study the robust utiltiy maximization problem, and its dual problem. 
In Section \ref{sec: examples} we provide three examples of specific models that are captured by our framework.
The proofs for our main results are given in the remaining sections. More precisely, the superhedging duality is proved in Section~\ref{sec: proof superhedging duality}. 
In Section~\ref{sec: proof separating duality} we show the duality between \(\cC\) and \(\cD\) while compactness and convexity of \(\cP\) and \(\cM\) is proved in Section~\ref{sec: compactness}.
In Section~\ref{sec: pf M = D} we verify that \(\cD = \cM\) in case a robust MPR exists.
The duality relation between \(u\) and \(v\) is established in Section~\ref{sec: proof duality utility}.


\newpage

\section{Main Result} \label{sec: results}
\subsection{The Setting}\label{sec: setting}
Fix a dimension \(d \in \mathbb{N}\) and a finite time horizon \(T > 0\). We define \(\Omega\) to be the space of continuous functions \([0,T] \to \mathbb{R}^d\) endowed with the uniform topology. 
The canonical process on $\Omega$ is denoted by \(X\), i.e., \(X_t (\omega) = \omega (t)\) for \(\omega \in \Omega\) and \(t \in [0,T]\). 
It is well-known that \(\mathcal{F} := \mathcal{B}(\Omega) = \sigma (X_t, t \in [0,T])\).
We define \(\F := (\mathcal{F}_t)_{t \in [0,T]}\) to be the canonical filtration generated by \(X\), i.e., \(\mathcal{F}_t := \sigma (X_s, s \in [0, t])\) for \(t \in [0,T]\). 
The set of probability measures on \((\Omega, \mathcal{F})\) is denoted by \(\mathfrak{P}(\Omega)\) and endowed with the usual topology of convergence in distribution.
Moreover, for any \(\sigma\)-field \(\cG\), let \(\cG^* := \bigcap_{P} \cG^P \) be the universal \(\sigma\)-field,
where \(P\) ranges over all probability measures on \(\cG\), and \(\cG^P\) denotes the completion of \(\cG\) w.r.t. \(P\).

Further, we denote the space of symmetric, positive semidefinite real-valued \(d\times d\) matrices by \(\mathbb{S}^d_{+}\), and by \(\bS^d_{++} \subset \bS^d_+\) the set of all positive definite matrices in \(\mathbb{S}^d_{+}\). Finally, recall that a subset of a Polish space is
called analytic if it is the image of a Borel subset of some Polish space
under a Borel map, and that a function \(f\) with values in \( \overline{\bR} := [-\infty, + \infty]\) is upper semianalytic if \(\{ f > c \}\) is analytic for every \(c \in \bR\). Any Borel function is also upper semianalytic.
We define, for two stopping times \(\rho\) and \(\tau\) with values in \( [0,T] \cup \{+\infty\}\), the stochastic interval
\[
\of \rho, \tau \of ~:= \{ (t,\omega) \in [0,T] \times \Omega \colon \rho(\omega) \leq t < \tau(\omega) \}.
\]
The stochastic intervals \( \gs \rho, \tau \of , \of \rho, \tau \gs , \gs \rho, \tau \gs  \) are defined accordingly.
In particular, the equality \( \of 0, T \gs = [0,T] \times \Omega \) holds.

Let \(F\) be a metrizable space and let \(b \colon F \times \of 0, T\gs \to \mathbb{R}^d\) and \(a \colon F \times \of 0, T \gs \to \mathbb{S}^d_{+}\) be two Borel functions such that \((t, \omega) \mapsto b(f, t, \omega)\) and \((t, \omega) \mapsto a (f, t, \omega)\) are predictable for all \(f \in F\).

We define the correspondences, i.e., the set-valued maps, \(\Theta, \tilde{\Theta} \colon \of 0, T \gs \hspace{0.05cm} \twoheadrightarrow \mathbb{R}^d \times \mathbb{S}^d_{+}\) by
\begin{align*}
    \Theta (t, \omega) &:= \big\{(b (f, t, \omega), a (f,t, \omega)) \colon f \in F \big\} \subset \mathbb{R}^d \times \mathbb{S}^d_{+},\\
    \tilde{\Theta} (t, \omega) &:= \{0\}^d \times \big\{a (f,t, \omega) \colon f \in F \big\} \subset \mathbb{R}^d \times \mathbb{S}^d_{+}.
\end{align*}

We denote the set of laws of continuous semimartingales by \(\fPs \subset \mathfrak{P}(\Omega)\).
For \(P \in \fPs\), we denote the semimartingale characteristics of the coordinate process \(X\) by \((B^P, C^P)\), and 
we set 
\[
\fPas  := \big\{ P \in \fPs  \colon P\text{-a.s. } (B^P, C^P) \ll \llambda \big\},
\]
where \(\llambda\) denotes the Lebesgue measure.

We further define, for fixed \(x_0 \in \bR^d\),
\begin{align*}
    \cP &:= \big\{ P \in \fPas \colon P \circ X_0^{-1} = \delta_{x_0},\ (\llambda \otimes P)\text{-a.e. } (dB^{P} /d\llambda, dC^{P}/d\llambda) \in \Theta   \big\} ,\\
    \cM &:= \big\{ Q \in \fPas \colon Q \circ X_0^{-1} = \delta_{x_0},\ (\llambda \otimes Q)\text{-a.e. } (dB^{Q} /d\llambda, dC^{Q}/d\llambda) \in \tilde{\Theta}   \big\}.
\end{align*}

\begin{SA} \label{SA: mbl and nonempty}
\quad
\begin{enumerate}
    \item[\textup{(i)}] \( \cP \neq \emptyset \neq \cM\).

    \item[\textup{(ii)}] \(\Theta\) and \(\tilde{\Theta}\) have a Borel measurable graph, i.e.,
\begin{align*}
\big\{ (t, \omega, b, a) \in [0, T] \times \Omega \times \mathbb{R}^d \times \mathbb{S}^d_+ \colon (b, a) \in \Theta (t, \omega) \big\} &\in \mathcal{B}([0, T]) \otimes \cF \otimes \mathcal{B}(\bR^d) \otimes \mathcal{B}(\bS^d_+), \\ 
\big\{ (t, \omega, b, a) \in [0, T] \times \Omega \times \mathbb{R}^d \times \mathbb{S}^d_+ \colon (b, a) \in \tilde{\Theta} (t, \omega) \big\} &\in \mathcal{B}([0, T]) \otimes \cF \otimes \mathcal{B}(\bR^d) \otimes \mathcal{B}(\bS^d_+).
\end{align*} 
\end{enumerate}
\end{SA}

\begin{remark} \label{rem: p and m nonempty}
\quad
\begin{enumerate} 
    \item[\textup{(i)}] By virtue of \cite[Lemma 2.10]{CN22}, part (i) from Standing Assumption \ref{SA: mbl and nonempty} holds if 
    the functions \(b \colon F \times \of 0, T\gs \to \mathbb{R}^d \) and \(a \colon F \times \of 0, T \gs \to \mathbb{S}^d_{+}\) are continuous and of linear growth, i.e.,
    there exists a constant \(\C > 0\) such that 
	\[
	\|b (f, t, \omega)\|^2 + \|a (f, t, \omega)\| \leq \C \Big ( 1       + \sup_{s \in [0, t]} \|\omega (s)\|^2 \Big )
	\]
	for all \((f, t, \omega) \in F \times \of 0, T\gs\).

    \item[\textup{(ii)}] Thanks to \cite[Lemma 2.8]{CN22}, part (ii) from Standing Assumption \ref{SA: mbl and nonempty} holds once \(F\) is compact and the functions \(b \colon F \times \of 0, T\gs \to \mathbb{R}^d \) and \(a \colon F \times \of 0, T \gs \to \mathbb{S}^d_{+}\) are continuous.
    This is a crucial property in order to use the theory developed in \cite{ElKa15, NVH}.
\end{enumerate}
\end{remark}

Let \(\cQ \subset \mathfrak{P}(\Omega)\) be a set of probability measures.
Recall that a \(\cQ\)-polar set is a set that is \(Q\)-null under every \(Q\in \cQ\), and 
that a property holds \(\cQ\)-quasi surely, if it holds outside a \(\cQ\)-polar set.
For two sets of probability measures \(\cQ, \cR \subset \mathfrak{P}(\Omega)\), we write
\( \cQ \ll \cR\) if every \(\cR\)-polar set is \(\cQ\)-polar. If \(\cQ \ll \cR\) and \(\cR \ll \cQ\),
we denote this by \(\cQ \sim \cR\).

Finally, for any collection \( \cR \subset \fPs\), we first define the filtration  \( \G^\cR =  (\cG^\cR_t)_{t \in [0, T]} \) via
\begin{equation} \label{eq: def G R}
\cG^\cR_t := \bigcap_{s > t} \big(\cF^*_s \vee \cN^{\cR}\big), \quad t \in [0, T], 
\end{equation}
where \( \cF^*_t \) is the universal completion of \( \cF_t \), and
\( \cN^{\cR}\) is the collection of \(\cR\)-polar sets. 
Then, we set
\(\cH^\cR\) to be the set of all \(\G^\cR\)-predictable processes \(H\) with \(H \in L(X, P)\)\footnote{See, for instance, \cite[Section III.6.c]{JS} for more details.}
for all \(P\in \cR\) and such that for every \(P \in \cR\) there exists a constant \(C = C(H, P) > 0\) such that \(P\)-a.s. \( \int_0^\cdot H_s d X_s \geq -C\). 

The following observation becomes useful later. If \(P \in \cR\) is such that \(X\) is a local \(P\)-martingale, then \(\int_0^\cdot H_s d X_s\) is a \(P\)-supermartingale for every \(H \in \cH^\cR\). This follows from the well-known fact that any local martingale that is bounded from below is a supermartingale.


\subsection{Superhedging duality} \label{sec: superhedging duality}
We denote the set of all local martingales measures for \(X\) that are absolutely continuous to the uncertainty set \(\mathcal{P}\) by
\[
 \mathfrak{M}_a (\cP) := \big\{ Q \in \fP_a(\cP) \colon X \text{ is a local \(Q\)-\(\F\)-martingale} \big\}, 
\]
where 
\[
\mathfrak{P}_a(\cP) := \big\{ Q \in \mathfrak{P}(\Omega) \colon \exists P \in \cP \text{ with } Q \ll P \big\}.
\]
The following theorem provides a version of \cite[Theorem 3.2]{nutz_superhedging} which is tailored to our nonlinear semimartingale framework.
It shows that for payoffs bounded from below, the optimal superhedging strategy is \emph{admissible} in a robust sense. This will turn out to be useful in Section~\ref{sec: proofs}.
The proof is given in Section~\ref{sec: proof superhedging duality} below.

\begin{theorem} \label{thm: superhedging duality}
Assume that \(\fM_a(\cP) \neq \emptyset\).
Let \( f \colon \Omega \to \bR_+\)  be an
upper semianalytic function such that
\[ \sup_{Q \in \mathfrak{M}_a(\cP)} E^Q\big[f\big] < \infty. \] 
Then, there exist a strategy \(H \in \cH^{\mathfrak{M}_a (\cP)}\) and a constant \(C > 0\) with
\[
\int_0^\cdot H_s d X_s \geq - C \quad  Q\text{-a.s. for all } Q \in \mathfrak{M}_a (\cP),
\]
such that
\[
\sup_{Q \in \mathfrak{M}_a (\cP)} E^Q \big[f \big] + \int_0^T H_s d X_s \geq f, \quad Q\text{-a.s. for all } Q \in \mathfrak{M}_a (\cP).
\]
\end{theorem}
To prove Theorem \ref{thm: superhedging duality} we have to verify the prerequisites of \cite[Theorem 3.2]{nutz_superhedging}.
To this end, in Section~\ref{sec: proof superhedging duality} we establish stability properties of \(\fM_a(\cP)\) that ensure the dynamic programming principle for the (dynamic) superhedging price associated to \(\fM_a(\cP)\).

\subsection{Separating duality for nonlinear continuous semimartingales}
\label{sec: separating duality}
For \(P \in \mathcal{P}\), we define 
\begin{align*}
    \mathfrak{M}_e^P := \big\{ Q \in \mathfrak{P} (\Omega) \colon P \sim Q, \  X \text{ is a local \(Q\)-\(\F\)-martingale} \big\}.
\end{align*}
Further, we denote the set of all local martingales measures for \(X\) that are equivalent to the uncertainty set \(\mathcal{P}\) by
\[
 \mathfrak{M}_e (\cP) := \big\{ Q \in \fP_e(\cP) \colon X \text{ is a local \(Q\)-\(\F\)-martingale} \big\}, 
\]
where 
\[
\mathfrak{P}_e(\cP) := \big\{ Q \in \mathfrak{P}(\Omega) \colon \exists P \in \cP \text{ with } Q \sim P \big\}.
\]

We define
\begin{align*}
    \cC & := \Big \{ g \colon \Omega \to [0, \infty] \colon  \text{ \(g\) is \(\cF_T^*\)-measurable and }  \exists H \in \cH^\cP \text{ with } 1 + \int_0^T H_s d X_s \geq g \ \cP\text{-q.s.} \Big\},
  \\
    \cD & := \big\{ Q \in \mathfrak{P}_a(\cP) \colon E^Q \big[ g \big] \leq 1 \text{ for all } g \in \cC \big\}.
\end{align*}
Here, \(\cC\) is the set of claims that can be \(\cP\)-quasi surely superreplicated  with initial capital \(1\), while 
\( \cD\) is the collection of absolutely continuous separating measures for \(\cC\).

In order to derive the duality between \(\cC\) and \(\cD\),
we impose the following \emph{no-arbitrage} condition.

\begin{definition} \label{def: NFLVR}
    We say that  \NFLVR holds if 
    for every real-world measure \(P \in \cP\) there exists a martingale measure \(Q \in \fM_a(\cP)\) such that \(P \ll Q\).
\end{definition} 

\begin{remark}
	If \NFLVR holds, then \(\cP\sim \fM_a (\cP)\). To see this, note that 
 \( \fM_a(\cP) \subset \fP_a(\cP) \ll \cP\) by definition. 
 Conversely, if \NFLVR holds, it follows that \(\cP \ll \fM_a(\cP)\).
\end{remark}

The following theorem can be viewed as a generalization of \cite[Propositions 5.7, 5.9]{kupper} from a L\'evy setting with constant \(\Theta (t, \omega) \equiv \Theta \subset \bR^d \times \mathbb{S}^d_{++}\) to a general nonlinear semimartingale framework with path dependent uncertainty set \((t, \omega) \mapsto \Theta (t, \omega)\). 
Moreover, it seems to be the first concrete result without an ellipticity assumption, i.e., which also covers incomplete market scenarios, see \cite[Remark 3.2]{kupper}. The theorem is proved in Section \ref{sec: proof separating duality} below.

\begin{theorem} \label{thm: dualities}
\quad 
\begin{enumerate}
\item[\textup{(i)}] It holds that
\begin{equation} \label{eq: first duality}
\mathfrak{M}_a (\cP) = \cD = \big\{ Q \in \mathfrak{P}_a(\cP) \colon E^Q \big[g \big] \leq 1 \text{ for all } g \in \cC \cap C_b  (\Omega; \mathbb{R})\big \}.   
\end{equation}
\item[\textup{(ii)}]
If \NFLVR holds, then 
\begin{equation} \label{eq: second duality}
 \cC \cap C_b (\Omega; \bR) = \big\{ g \in C_b^+(\Omega; \bR) \colon E^Q \big[g\big] \leq 1 \text{ for all } Q \in \cD\big\}.    
\end{equation}
\end{enumerate}
\end{theorem}

\begin{remark}[Discussion of  \NFLVR] \label{rem: NFLVR}
Notice that in the single-measure case \(\cP = \{P\}\), \NFLVR is equivalent 
to the existence of an \emph{equivalent} local martingale measure \(Q \in \fM_e^P\). 
Thanks to the seminal work of Delbaen and Schachermayer (cf. \cite{DS} for an overview), the existence of an equivalent local martingale measure is equivalent to the absence of arbitrage in the sense of {\em no free lunch with vanishing risk (NFLVR)}.
Further, observe that \NFLVR is implied by the existence of a measure \(Q \in \fM_e^P\) for every \(P \in \cP\), i.e., if the NFLVR conditions holds under every \(P \in \cP\).
In general, \NFLVR is the continuous time version of the robust no-arbitrage condition NA(\(\cP\)) from \cite{nutz_nondom}.
More precisely, in finite discrete time, \cite[Theorem 4.5]{nutz_nondom} shows that NA(\(\cP\)) is equivalent to both, the 
mere equivalence of \(\fM_a(\cP)\) and \(\cP\), and the seemingly stronger condition that for every \(P \in \cP\) there exists \(Q \in \fM_a(\cP)\) such that \(P \ll Q\).
We impose the latter condition to account for continuous time stochastic integration. It guarantees that \(\fM_a(\cP)\) is sufficiently rich in the sense that it implies the equality
\(\cH^\cP = \cH^{\fM_a(\cP)}\) (see Lemma \ref{lem: integrands} below). This equality is crucial in our proof of the separating duality.
\end{remark}

The dualities \eqref{eq: first duality} and \eqref{eq: second duality} are two of the three main hypothesis from \cite{kupper}. The third main assumption translates to compactness and convexity of \(\cP\) and \(\mathfrak{M}_a (\cP)\). 

Next, we investigate compactness and convexity of \(\cP\) and \(\fM_a(\cP) = \cD\). To treat the second set, we prove that it coincides with \(\cM\) under the existence of a robust market price of risk, which is a notion we introduce in Condition \ref{SA: MPR} below.

Before, we formulate parameterized conditions that ensure compactness and convexity of \(\cP\) and \(\cM\), respectively.

\begin{condition} \label{cond: compact, LG, cont}
\quad 
\begin{enumerate}
\item[\textup{(i)}] \(F\) is compact.
\item[\textup{(ii)}]
The functions \(b \colon F \times \of 0, T\gs \to \mathbb{R}^d\) and \(a \colon F \times \of 0, T \gs \to \mathbb{S}^d_{+}\) are continuous.
\item[\textup{(iii)}]
There exists a constant \(\C > 0\) such that 
	\[
	\|b (f, t, \omega)\|^2 + \|a (f, t, \omega)\| \leq \C \Big ( 1 + \sup_{s \in [0, t]} \|\omega (s)\|^2 \Big )
	\]
	for all \((f, t, \omega) \in F \times \of 0, T\gs\).
\end{enumerate}
\end{condition}

\begin{condition} \label{cond: convex}
    The correspondence \(\Theta\) is convex-valued, i.e., \(\{(b (f, t, \omega), a(f, t, \omega )) \colon f \in F\} \subset \bR^d \times \mathbb{S}^d_+\) is convex for every \((t, \omega) \in \of 0, T\gs\).
\end{condition}

Observe that Condition \ref{cond: convex} also ensures that \(\tilde{\Theta}\) is convex-valued.
Additionally, recall that a continuous semimartingale is a local martingale if and only if its first characteristic vanishes. Hence, \(\cM\) consists of local martingale measures. In fact, by the linear growth condition (iii) from Condition \ref{cond: compact, LG, cont}, the set \(\cM\) then even consists of (true) martingale measures.

The following theorem was established in \cite[Propositions~3.9, 5.7]{CN22b} for one-dimensional nonlinear diffusions, but the argument transfers to our multidimensional path dependent framework. 
Its compactness part extends \cite[Theorem 4.41]{hol16} and \cite[Theorem~2.5]{neufeld} beyond the case where \(b\) and \(a\) are of Markovian structure, uniformly bounded and globally Lipschitz continuous.\footnote{After submitting this paper, we established a more general version of Theorem \ref{thm: compactness} in the (updated) paper \cite{CN22}.} The following result is proved in Section \ref{sec: compactness} below.

\begin{theorem} \label{thm: compactness}
Suppose that the Conditions \ref{cond: compact, LG, cont} and \ref{cond: convex} hold. Then, the sets \( \cP \) and \( \cM\) are both convex and compact.
\end{theorem}

In the final part of this section, we provide a condition under that \(\mathfrak{M}_a (\cP) = \cM\). 
This, together with Theorem \ref{thm: compactness}, will provide compactness of \(\mathfrak{M}_a (\cP)\). Compactness is a key ingredient needed for the robust duality theory developed in Section \ref{sec: duality utility} below.

\begin{condition}[Existence of robust market price of risk (MPR)] \label{SA: MPR}
There exists a Borel function \(\theta \colon F \times \of 0, T\gs \to \mathbb{R}^d\) such that 
\(
b = a \theta.
\)
Moreover, for every \(N > 0\), there exists a constant \(C = C_N > 0\) such that 
\[
\sup \big\{ (\langle \theta, a \theta \rangle) (f, t, \omega)  \colon f \in F, t < T_N (\omega) \big\} \leq C 
\]
for all \(\omega \in \Omega\), where \(T_N (\omega) := \inf \{t \in [0, T] \colon \|\omega (t)\| \geq N \} \wedge T\).
\end{condition}

In order to apply the results of \cite{kupper}, it remains to verify \cite[Assumption 2.1]{kupper}, i.e., that
\begin{equation} \label{eq: cond 2.1}
\text{for every \(P \in \cP\) there exists \(Q \in \fM_a(\cP)\) such that \(Q \ll P\).}   
\end{equation}
By means of Condition \ref{SA: MPR}, we establish in Theorem \ref{thm: M = D} below the stronger statement that for every \(P \in \cP\) there exists a measure \(Q \in \fM_e(\cP)\) such that \(Q \sim P\). Theorem \ref{thm: M = D} constitutes the last main result of this section. Its proof is given in Section \ref{sec: pf M = D}.

\begin{theorem} \label{thm: M = D}
Assume that the Conditions \ref{cond: compact, LG, cont} and \ref{SA: MPR} hold.
Then,
\begin{enumerate}
    \item[\textup{(i)}] for every \( P \in \cP\), there exists a measure \(Q \in \cM\) with \(P \sim Q\), and

    \item[\textup{(ii)}] for every \( Q \in \cM\), there exists a measure \(P \in \cP\) with \( Q \sim P\).
\end{enumerate}
\end{theorem}

\begin{corollary} \label{cor: M = D}
Assume that the Conditions \ref{cond: compact, LG, cont} and \ref{SA: MPR} hold. Then,
\begin{enumerate}
    \item[\textup{(i)}] the equalities \(\cM = \mathfrak{M}_e(\cP) = \mathfrak{M}_a(\cP)\) hold, and

    \item[\textup{(ii)}] for every \(P \in \cP\) there exists a measure \(Q \in \fM_e^P\).
\end{enumerate}
In particular, \NFLVR holds.
\end{corollary} 

\begin{proof}
   Notice that, by Girsanov's theorem (\cite[Theorem~III.3.24]{JS}), 
   \[
   \fM_e(\cP) = \cM \cap \fP_e(\cP) \text{ and }  \fM_a(\cP) = \cM \cap \fP_a(\cP).
   \]
   The second part of Theorem \ref{thm: M = D} implies
   \( \cM \subset \fP_e(\cP)\) and therefore \( \fM_e(\cP) = \cM\).
   Hence, we conclude that
   \[
   \fM_a(\cP) \subset \cM = \fM_e(\cP) \subset \fM_a(\cP),
   \]
   which shows the first statement.
   Further, the first part of Theorem \ref{thm: M = D} shows that for every \(P \in \cP\) there exists a measure \(Q \in \fM_e^P\).
\end{proof}

\begin{remark}[Compactness of \(\fM_e(\cP)\)]
    Together with Theorem \ref{thm: compactness}, Corollary \ref{cor: M = D} shows that \(\fM_e(\cP) = \fM_a(\cP)\) is compact. In the single measure case \(\cP = \{P\}\), compactness
    of \(\fM_e(\cP) = \fM_e^P\) is equivalent to  \(\fM_e^P\) being a singleton (given it is nonempty), i.e., the market is complete. Indeed, in case there exist two distinct elements in \(\fM_e^P\), George Lowther \cite{GL} has shown that there exists a local martingale measure \(Q\) absolutely continuous with respect to \(P\) but \(Q \not \in \fM_e^P\). In particular, \(\fM_e^P\) then fails to be closed.
    In the robust case, \(\fM_e(\cP)\) can be compact without the necessity that every physical measure \(P \in \cP\) corresponds to a complete market model.
\end{remark}

It seems to us that the structure of Condition \ref{SA: MPR} is new in the literature on robust dualities. 
In the following we discuss the relation of Condition \ref{SA: MPR} to the classical notion of a MPR and \NFLVR. Further, we relate Condition \ref{SA: MPR} to ellipticity conditions that have previously appeared in the literature, and we show with an example that \(\cM= \mathfrak{M}_a (\cP)\) fails without Condition~\ref{SA: MPR}.

\begin{remark}[Relation to classical MPRs] \label{rem: martingale measures}
  Let us shortly explain why \(\theta\) from Condition \ref{SA: MPR} can be considered as a {\em robust version} of a MPR. Take a real-world measure \(P \in \cP\) and denote the Lebesgue densities of the \(P\)-characteristics of \(X\) by \((b^P, a^P)\). 
    Under Condition \ref{SA: MPR}, in Lemma \ref{lem: MPR meas selection} below we establish the existence of a predictable function \(\f = \f (P) \colon \of 0, T \gs \to F\) such that, for \((\llambda \otimes P)\)-a.a. \((t, \omega) \in \of 0, T\gs\),  
    \[
    (b^P_t (\omega), a^P_t (\omega)) = (b(\f (t, \omega), t, \omega), a (\f (t, \omega), t, \omega)) = (a (\f (t, \omega), t, \omega) \theta (\f (t, \omega), t, \omega), a (\f (t, \omega), t, \omega)).\] 
    This representation shows that \((t, \omega) \mapsto \theta (\f (t, \omega), t, \omega)\) is a MPR in the classical sense for the real-world measure \(P\). 
    As \(P \in \cP\) was arbitrary, this leads to our interpretation of \(\theta\) as a {\em robust version} of a MPR. 
\end{remark}    

\begin{remark}[Relation to \NFLVR]
As observed in the seminal paper \cite{ChuStr}, the existence of a MPR is equivalent to the {\em no unbounded profits with bounded risk (NUPBR)} condition that has been introduced in \cite{KarKar}. In the context of utility maximization, assuming NUPBR is very natural, as it is known to be the minimal a priori assumption needed to proceed with utility optimization, see \cite{KarKar}. 
We stress that Condition~\ref{SA: MPR} is in fact a bit more than only the existence of a MPR. Indeed, we require in addition some mild local boundedness property, which is of technical nature.     
The difference between NUPBR and NFLVR can be captured by the martingale property of a non-negative local martingale (a so-called {\em strict local martingale deflator}). In our setting, we establish such martingale properties with help of the linear growth condition (iii) from Condition \ref{cond: compact, LG, cont}. 
This allows us to verify that for every model \(P \in \cP\) the NFLVR condition holds.
\end{remark}

\begin{remark}[Relation to ellipticity] \label{rem: elliptic}
The results from \cite{kupper,denis} require a uniform ellipticity assumption.
Let us briefly explain how our framework and particularly Condition \ref{SA: MPR} relates to the setting from~\cite{kupper}.
As in \cite{kupper}, we take a compact and convex set \(\Theta \subset \bR^d \times \bS^d_+\). The uniform ellipticity assumption \cite[Assumption~3.1]{kupper} reads as follows:
\begin{equation} \label{eq: kupper ellipticity}
\text{there exists a matrix } \underline{A} \in \bS^d_{++}
\text{ such that } \underline{A} \leq A,
\quad A \in \Theta_2,
\end{equation}
where 
\[ \Theta_2 := \{ A \in \bS^d_+ \colon \exists B \in \bR^d \text{ with } (B,A) \in \Theta \},
\]
and \(\underline{A} \leq A \) means that \(A - \underline{A} \in \bS^d_{+}\).
Notice that \eqref{eq: kupper ellipticity}, together with the boundedness of \(\Theta_2\), is equivalent to the existence of a constant \(\K > 0 \) such that, for all \(\xi \in \mathbb{R}^d\) with \(\|\xi\| = 1\),
\[
\frac{1}{\K} \leq \langle \xi, A \xi \rangle   \leq \K, 
\quad A \in \Theta_2.
\]

To wit, in our notation this can be recovered by setting 
\(F := \Theta\) and defining the functions 
\(b  \colon F \to \bR^d\) and \(a \colon F \to \bS^d_+\)
as the projections on the first and second coordinate, respectively. Then, \eqref{eq: kupper ellipticity} translates to
\begin{equation} \label{eq: kupper ellipticity new}
    \text{there exists a matrix } \underline{A} \in \bS^d_{++} \text{ such that } \underline{A} \leq a(f), \quad f \in F.  
\end{equation}

In this case, we can decompose \(b = a \theta\) with \(\theta = a^{-1} b\). If, in addition, the drift \(b\) is uniformly bounded as in \cite[Assumption 3.1]{kupper}, the (unique) market price of risk \(\theta\) is (globally) bounded and Condition \ref{SA: MPR} holds.
\end{remark}

\begin{remark}[\(\cM = \mathfrak{M}_a (\cP)\) fails without Condition \ref{SA: MPR}] \label{rem: MPR discussion}
The existence of a MPR is necessary for the identity \(\cM = \mathfrak{M}_a (\cP)\) to hold. We now give an example where \(\fM_a (\cP) \subsetneq \cM\) and a MPR exists but it fails to be locally bounded.

Let \(d =1\) and \(x_0 > 0\), and take \(F : = [1, 2] \times [1, 2]\) and 
\[b ((f_1, f_2), t, \omega) := f_1 \cdot |\omega (t)|^{1/2}, \qquad a ((f_1, f_2), t, \omega) := f_2 \cdot | \omega (t) |^{3/2},\] for \(((f_1, f_2), t, \omega) \in F \times \of 0, T\gs\). 
Evidently, Condition \ref{cond: compact, LG, cont} is satisfied.
Moreover, 
\[
\theta ((f_1, f_2), t, \omega) := \frac{f_1}{f_2} \frac{\1_{\{\omega (t) \not = 0\}}}{|\omega (t)|}, \quad ((f_1, f_2), t, \omega) \in F \times \of 0, T\gs, 
\]
is a MPR but it fails to satisfy the local boundedness assumption from Condition \ref{SA: MPR}.
We now prove that \(\fM_a (\cP) \subsetneq \cM\).
Let \(Q\) be a solution measure (i.e., the law of a solution process) for the stochastic differential equation (SDE)
\[
d Y_t = |Y_t|^{3/4} d W_t, \quad Y_0 = x_0, 
\]
where \(W\) is a one-dimensional standard Brownian motion. Such a measure \(Q\) exists by Skorokhod's existence theorem (see, e.g., \cite[Theorem 4, p. 265]{skorokhod}). Furthermore, it is clear that \(Q \in \mathcal{M}\). In the following we show that \(Q \not \in \fM_a (\mathcal{P})\). 
Set \(T_0 := \inf \{t \in [0, T] \colon X_t = 0\}\). As \[
\int_0^1 \frac{x \hspace{0.03cm}dx}{|x|^{3/4 \hspace{0.02cm} \cdot \hspace{0.02cm} 2}} = \int_0^1 \frac{dx}{x^{1/2}} = 2 < \infty,
\]
we deduce from Feller's test for explosion (cf., e.g., \cite[Proposition 2.12]{MU12ECP}) and \cite[Theorem~1.1]{BR16} that \[Q (T_0 < T) > 0.\] 
Set \(\underline{b} (x) := 1/(2x), \overline{b} (x) := 2 /x\) and \(\overline{a} (x) := 2 x^{3/2}\) for \(x > 0\). Furthermore, for suitable Borel functions \(f\colon (0, \infty) \to \bR\) and \(g \colon (0, \infty) \to (0, \infty)\), define 
\[
v (f, g) (x) := \int_1^x \exp \Big( - \int_1^y 2 f (z) dz \Big) \int_1^y \frac{2 \exp (\int_1^\xi 2 f (z) dz )}{g (\xi)} d\xi dy, \quad x > 0.
\]
Notice that 
\begin{equation} \label{eq: FT}
\begin{split}
v (\underline{b}, \overline{a}) (x) &= 4 \big[ x^{1/2} - 1\big] - 2 \log (x) \to \infty, \ \ x \searrow 0, 
\\
v (\overline{b}, \overline{a}) (x) &= \tfrac{2}{21} \big[ x^{-3} + 6 x^{1/2} - 7 \big] \to \infty, \ \ x \nearrow \infty.
\end{split}
\end{equation}
Take a measure \(P \in \mathcal{P}\). By definition, we have \(P\)-a.s. for \(\llambda\)-a.e. \(t < T_0\)
\[
d C^P_t / d \llambda \leq \overline{a} (X_t), \qquad d C^P_t / d \llambda \cdot \underline{b} (X_t) \leq d B^P_t /d \llambda \leq d C^P_t / d \llambda \cdot \overline{b} (X_t).
\]
Hence, taking \eqref{eq: FT} into account, we deduce from \cite[Theorem 5.2]{C20} that \(P (T_0 = \infty) = 1\).
In summary, \(Q (T_0 < T) > 0\) and \(P (T_0 < T) = 0\). As \(P\) was arbitrary, we conclude that \(Q \not \in \fM_a (\cP)\).
\end{remark}

\subsection{Duality theory for robust utility maximization} \label{sec: duality utility}
Let \( U \colon (0, \infty) \to [-\infty, \infty) \) be a utility function, i.e., a concave and non-decreasing function. We define \( U(0) := \lim_{x \searrow 0} U(x) \), and consider the conjugate function
\begin{align*}
    V(y) := \begin{cases} \sup_{x \geq 0} [U(x)-xy],& y > 0, \\
     \lim_{y \searrow 0} V(y),& y = 0, \\
   \infty,&y < 0. \end{cases} 
\end{align*}

We set, for \(x, y > 0\),
\begin{align*}
    \cC(x) & := x \hspace{0.05cm} \cC ,\qquad
    \cD(y)  := y\hspace{0.05cm} \cD,\\
\end{align*}
and
\begin{align*}
    u(x) & :=  \sup_{g \in \cC(x)} \inf_{P \in \cP} E^P \big [U(g) \big], \qquad  v(y) := \inf_{Q \in \cD(y)} \sup_{P \in \cP} E^P\Big[V\Big(\frac{dQ}{dP}\Big) \Big], 
\end{align*} 
with the convention \( \frac{dQ}{dP} := - \infty \) in case \(Q\) is not absolutely continuous with respect to \(P\).

The separating dualities in Theorem \ref{thm: dualities} allow us to establish a conjugacy relation between \(u\) and \(v\) for utility functions bounded from below. 
This is in the spirit of \cite[Theorem 2.1]{kramkov}, as it only requires a no-arbitrage assumption in the sense of Condition \ref{SA: MPR} and finiteness of the value function \(u\).

\begin{theorem} \label{thm: conjugacy bounded}
Assume that the Conditions \ref{cond: compact, LG, cont}, \ref{cond: convex} and \ref{SA: MPR} hold.
Let \(U\) be a utility function with \(U(0) > - \infty\) and assume that \( u(x_0) < \infty\) for some \( x_0 > 0\). Then,
\begin{enumerate}
\item[\textup{(i)}] \(u\) is nondecreasing, concave and real-valued on \( (0, \infty)\),
\item[\textup{(ii)}] \( v\) is nonincreasing, convex, and proper,
\item[\textup{(iii)}] the functions \(u\) and \(v\) are conjugates, i.e., 
\[
u(x) = \inf_{y > 0} \big[v(y) + xy \big], \ \ x > 0, \quad v(y) = \sup_{x > 0} \big[u(x) - xy \big], \ \ y > 0,
\]
\item[\textup{(iv)}] for every \(x > 0\) we have
\[
u(x) = \sup_{g \in \cC(x)} \inf_{P \in \cP} E^P \big[U(g) \big] = \sup_{g \in \cC(x)\cap C_b} \inf_{P\in\cP} E^P \big[U(g) \big].
\]
\end{enumerate}
\end{theorem}

\begin{proof}
Corollary \ref{cor: M = D} implies that \NFLVR holds. Hence, we deduce from Theorem \ref{thm: dualities} that \(\cC\) and \(\cD\) are in duality and that \(\cD = \fM_a(\cP)\).
Applying Corollary \ref{cor: M = D} once more, Theorem \ref{thm: compactness} shows that the sets 
\(\cP\) and \(\cD = \cM\) are convex and compact.
Using Corollary \ref{cor: M = D} a third time proves that \eqref{eq: cond 2.1} holds.
Thus, the claim follows from \cite[Theorem 2.10]{kupper}.
\end{proof}

We now aim at the extension of Theorem \ref{thm: conjugacy bounded} regarding the existence of an optimal portfolio and to utilities unbounded from below.
The paper \cite{kupper} provides abstract conditions that guarantee both, the existence of an optimal (generalized) portfolio and an extension of Theorem \ref{thm: conjugacy bounded}.
In order to give verifiable parameterized condition in terms of \(b\) and \(a\), we now focus on specific utilities.
That is, \(U\) is assumed to be one of the following
\[
U(x) = \log(x), \quad U(x) = -e^{-\lambda x} \text{ for } \lambda > 0, \quad U(x) = \frac{x^p}{p} \text{ for } p \in (-\infty,0) \cup (0,1). 
\]

The next two conditions are used when we consider utility functions unbounded from above, i.e., in case of logarithmic and power utility with parameter \(p \in (0,1)\). 
We will only require that \emph{one} of them holds.
They provide sufficient integrability of the utility of portfolios and thus entail, in particular, finiteness of the value function \(u\).

\begin{condition} \label{cond: MPR bounded} 
There exists a robust market price of risk \(\theta\) as in Condition \ref{SA: MPR} and
the function \(\langle \theta, a \theta \rangle\) is uniformly bounded.
\end{condition}

\begin{condition}[Uniform ellipticity and boundedness of volatility] \label{cond: unif ellipticity vola}
There exists a constant \(\K > 0\) such that, for all \(\xi \in \mathbb{R}^d\) with \(\|\xi\| = 1\),
\[
\frac{1}{\K} \leq \langle \xi, a (f, t, \omega) \xi \rangle \leq \K
\]
for all \((f,t, \omega) \in F \times \of 0, T\gs\).
\end{condition}

\begin{remark}
From a technical point of view, we require one of the Conditions~\ref{cond: MPR bounded} and~\ref{cond: unif ellipticity vola} to obtain that certain stochastic exponentials have polynomial moments, which are needed to treat unbounded utility functions.
Such moments are readily established under the global boundedness assumption from Condition \ref{cond: MPR bounded} but less obvious under Condition \ref{cond: unif ellipticity vola}.
We give a precise statement in Proposition \ref{prop: integrability of stochastic exponential} below that we believe to be of independent interest.

We emphasise that the scope of the assumptions is different. On one hand, Condition~\ref{cond: MPR bounded} is a boundedness condition but it covers incomplete market models. On the other hand, Condition~\ref{cond: unif ellipticity vola} allows the drift coefficient \(b\) to be unbounded but it enforces a robust version of market completeness. 

Both conditions are satisfied for the L\'evy framework from \cite[Section~3]{kupper} that was discussed in Remark~\ref{rem: elliptic}.
Beyond the L\'evy case, Condition \ref{cond: unif ellipticity vola} holds e.g. for real-valued nonlinear affine processes as studied in \cite{fadina2019affine}. In general, ellipticity assumptions (not necessarily uniform) are quite standard in the literature on linear multidimensional diffusions (see, e.g., \cite{SV}).
\end{remark}

Next, we explain a suitable notion of generalized portfolios. 

\begin{remark}[Existence of an optimal portfolio] 
To show the existence of an optimal portfolio \(g^* \in \cC(x)\) with
\[
u(x) = \sup_{g \in \cC(x)} \inf_{P \in \cP} E^P \big[U(g) \big] = \inf_{P \in \cP} E^P \big[U(g^*) \big],
\]
we cannot rely on classical arguments such as Koml{\'o}s' lemma.
In the realm of robust finance, medial limits have been a suitable substitute
and allow us to construct a (generalized) optimal portfolio as the medial limit of a sequence of near-optimal portfolios.

Here,
a \emph{medial limit} is a positive linear functional \( \limmed \hspace{-0.075cm} \colon \ell^\infty \to \bR\)
satisfying 
\[
\liminf_{n \to \infty} x_n \leq \underset{n \to \infty}\limmed x_n \leq \limsup_{n \to \infty} x_n
\]
for every bounded sequence \( (x_n)_{n \in \bN} \in \ell^\infty\), and for every bounded sequence \( (X_n)_{n\in\bN} \)
of universally measurable functions \( X_n \colon \Omega \to \bR \), the medial limit \( \limmed_{n \to \infty} X_n \)
is again universally measurable with
\[
E^P\Big[\underset{n \to \infty}\limmed X_n \Big] = \underset{n \to \infty}\limmed E^P[X_n],
\]
for every \(P \in \mathfrak{P}(\Omega)\). Note that a medial limit extends naturally to sequences with values in \([-\infty, \infty]\).

We refer to \cite{bartl, kupper, nutz_integrals, nutz_superreplication} for recent applications of medial limits.
Let us also mention that the existence of medial limits is guaranteed when working under ZFC together with Martin's axiom, see \cite{meyer, normann}. 
\end{remark}

\begin{condition} \label{cond: medial limit}
Medial limits exist.
\end{condition}

Provided Condition \ref{cond: medial limit} is in force, we set
\[
    \overline{\cC}  := \big\{ g \colon \Omega \to [0,\infty] \colon g =  \underset{n \to \infty}\limmed\hspace{0.025cm} g_n \text{ with } (g_n)_{n \in \bN} \subset \cC \big\},
\]
and, for \(x, y > 0\), we define
\begin{align*}
   \overline{\cC}(x)  := x \hspace{0.05cm} \overline{\cC} ,\qquad
   \overline{u}(x) := \sup_{g \in \overline{\cC}(x)} \inf_{P \in \cP} E^P \big[U(g) \big] .
\end{align*}

We are in the position to state the main results of this paper. The first of the following two theorems deals with utility functions which are bounded from below and the second deals with utility functions which are unbounded from below. 
Note that in case the utility is unbounded from below we have to rely on generalized portfolios.

The following theorems extend \cite[Theorems 3.4, 3.5]{kupper} from a L\'evy setting to a general nonlinear semimartingale framework. 

\begin{theorem} \label{thm: main positive power exponential}
Assume that the Conditions \ref{cond: compact, LG, cont}, \ref{cond: convex} and  \ref{SA: MPR} hold.
Let \(U\) be either a power utility \(U(x) = x^p / p\) with exponent \( p \in (0,1)\),
or an exponential utility \(U (x) = - e^{- \lambda x}\) with parameter \( \lambda > 0\). In case \(U\) is a power utility assume also that either Condition \ref{cond: MPR bounded} or Condition \ref{cond: unif ellipticity vola} holds. Then,
\begin{enumerate}
\item[\textup{(i)}] \(u\) is nondecreasing, concave and real-valued on \( (0, \infty)\),
\item[\textup{(ii)}] \( v\) is nonincreasing, convex, and proper,
\item[\textup{(iii)}] the functions \(u\) and \(v\) are conjugates, i.e., 
\[
u(x) = \inf_{y > 0} \big[v(y) + xy \big], \ \ x > 0, \quad v(y) = \sup_{x > 0} \big[u(x) - xy \big], \ \ y > 0,
\]
\item[\textup{(iv)}] for every \(x > 0\) we have
\[
u(x) = \sup_{g \in \cC(x)} \inf_{P \in \cP} E^P \big[U(g) \big] = \sup_{g \in \cC(x)\cap C_b} \inf_{P\in\cP} E^P \big[U(g) \big].
\]
\end{enumerate}
If additionally Condition \ref{cond: medial limit} holds, then
\begin{enumerate}
\item[\textup{(v)}]  for every \( x > 0 \) we have \( u(x) = \overline{u}(x)\),
\item[\textup{(vi)}] for every \( x > 0\) there exists \( g^* \in \overline{\cC}(x)\) such that
\[
\overline{u}(x) = \sup_{g \in \overline{\cC}(x)} \inf_{P \in \cP} E^P \big [U(g) \big] = \inf_{P\in\cP} E^P \big [U(g^*) \big].
\]
\end{enumerate}
\end{theorem}

\begin{theorem} \label{thm: main negative power log}
Assume that the Conditions \ref{cond: compact, LG, cont}, \ref{cond: convex}, \ref{SA: MPR} and \ref{cond: medial limit} hold.
Let \(U\) be either a power utility \(U (x) = x^p / p\) with exponent \( p \in (-\infty,0)\),
or the log utility \(U(x) = \log (x)\). In case \(U\) is the log utility assume also that either Condition \ref{cond: MPR bounded} or Condition \ref{cond: unif ellipticity vola} holds.
Then,
\begin{enumerate}
\item[\textup{(i)}] \(\overline{u}\) is nondecreasing, concave and real-valued on \( (0, \infty)\),
\item[\textup{(ii)}] \( v\) is nonincreasing, convex, and real-valued on \( (0, \infty)\),
\item[\textup{(iii)}] the functions \(\overline{u}\) and \(v\) are conjugates, i.e., 
\[
\overline{u}(x) = \inf_{y > 0} \big [v(y) + xy \big], \ \ x > 0, \quad v(y) = \sup_{x > 0} \big[\overline{u}(x) - xy \big], \ \ y > 0,
\]
\item[\textup{(iv)}] for every \(x > 0\) there exists \(g^* \in \overline{\cC}(x)\) with
\[
\overline{u}(x) = \sup_{g \in \overline{\cC}(x)} \inf_{P \in \cP} E^P \big[U(g) \big] = \inf_{P\in\cP} E^P\big[U(g^*) \big].
\]
\end{enumerate}
\end{theorem}

The proofs of the Theorems \ref{thm: main positive power exponential} and \ref{thm: main negative power log} are given in Section \ref{sec: proof duality utility} below. In the next section we
discuss a variety of frameworks to which our results apply.

\section{Examples} \label{sec: examples}
In this section we mention three examples of stochastic models that are covered by our framework. We stress that it includes many previously studied frameworks but also some new ones which are of interest for future investigations.

\subsection{Nonlinear diffusions} \label{sec: nonlinear diffusions}
Let \(b' \colon F \times \bR^d \to \bR^d\) and \(a' \colon F \times \bR^d \to \mathbb{S}^d_+\) be two Borel functions and set, for 
\((f, t, \omega) \in F \times \of 0, T\gs\),
\[
b (f, t, \omega) = b' (f, \omega (t)), \quad a (f, t, \omega) = a' (f, \omega (t)).
\]
This setting corresponds to a \emph{nonlinear diffusion framework}. In particular, in this case the correspondences \( \Theta\) and \(\tilde{\Theta} \) have \emph{Markovian structure}, i.e., the sets \( \Theta(t, \omega), \tilde{\Theta}(t,\omega) \) depend on \((t, \omega)\) only through the value \( \omega(t) \).
To provide some further insights, we require additional notation.
For each \( x \in \bR^d \), we set 
 \[
\cR (x) := \big\{ P \in \fPas \colon P \circ X_0^{-1} = \delta_{x},\ (\llambda \otimes P)\text{-a.e. } (dB^{P} /d\llambda, dC^{P}/d\llambda) \in \Theta   \big\},
\]
and further define the sublinear operator \( \cE^x \) on the convex cone of upper semianalytic functions \(\psi \colon \Omega \to [- \infty, \infty] \) by 
\[ \cE^x(\psi) := \sup_{P \in \cR(x)} E^P \big[ \psi \big]. \] 
For every \( x \in \bR^d \), we have by construction that \( \cE^x(\psi(X_0)) = \psi(x) \) for every upper semianalytic function \( \psi \colon \bR^d \to \bR \).

Denote, for \(t \in [0, T]\), the shift operator \(\theta_t \colon \Omega \to \Omega\) by \(\theta_t (\omega) := \omega((\cdot + t) \wedge T)\) for all \(\omega \in \Omega\). As in \cite[Proposition~ 2.8]{CN22b}, we obtain the following result.

\begin{proposition} \label{prop: markov property}
For every upper semianalytic function \( \psi \colon \Omega \to \bR \), the equality
\[
\cE^x( \psi \circ \theta_t) = \cE^x ( \cE^{X_t} (\psi))
\]
holds for every \(x \in \bR^d \), and every \( t \in [0, T] \).
\end{proposition}

Proposition \ref{prop: markov property} confirms the intuition that the coordinate process is a nonlinear Markov process under the family \( \{\cE^x \colon x \in \bR\} \), as it implies the equality
\[ \cE^x(\psi(X_{s+t})) = \cE^x( \cE^{X_t}( \psi (X_s))) \]
for every upper semianalytic function \(\psi \colon \bR^d \to \mathbb{R}\), \(s,t \in [0, T] \) with \(s + t \leq T\), and \(x \in \bR^d\).

Notice that the Conditions \ref{cond: compact, LG, cont} and \ref{cond: convex} are implied by the following conditions:
\begin{enumerate}
\item[\textup{(i)\('\)}] \(F\) is compact.
\item[\textup{(ii)\('\)}] \(b'\) and \(a'\) are continuous.
\item[\textup{(iii)\('\)}] There exists a constant \(\C > 0\) such that 
	\[
	\|b' (f, x)\|^2 + \|a' (f, x)\| \leq \C \big ( 1 + \|x\|^2 \big )
	\]
	for all \((f, x) \in F \times \bR^d\).
\item[\textup{(iv)\('\)}] The set \(\{(b' (f, x), a' (f, x)) \colon f \in F\} \subset \bR^d \times \mathbb{S}^d_+\) is convex for every \(x \in \bR^d\).
\end{enumerate}

Summing up, this setting provides a robust counterpart to classical continuous Markovian financial frameworks. In particular, it allows to combine different Markovian models such as, for instance, Cox--Ingersoll--Ross and Vasi{\u c}ek models. 

\subsection{Random generalized \(G\)-Brownian motions}
An economically interesting situation, previously studied in \cite{NVH}, see also \cite{nutz}, is the case where \(d = 1\) and, for \((t, \omega) \in \of 0, T\gs\),
\begin{align*}
\Theta (t, \omega) &:= [\underline{b}_t (\omega), \overline{b}_t (\omega)] \times [\underline{a}_t (\omega), \overline{a}_t (\omega)], \\
\tilde{\Theta} (t, \omega) &:= \{0\} \times [\underline{a}_t (\omega), \overline{a}_t (\omega)],
\end{align*}
where \(\underline{b}, \overline{b} \colon \of 0, T \gs \to \bR\) and \(\underline{a}, \overline{a} \colon \of 0, T \gs \to \bR_+\) are predictable functions such that
\[
\underline{b} \leq \overline{b}, \quad \underline{a} \leq \overline{a}.
\]
In this case, the sets \(\mathcal{P}\) and \(\mathcal{M}\) are given by 
\begin{align*}
    \mathcal{P} &= \big\{ P \in \fPas \colon P \circ X_0^{-1} = \delta_{x_0}, \ (\llambda \otimes P)\text{-a.e. } \underline{b} \leq d B^P / d \llambda \leq \overline{b}, \ \underline{a} \leq d C^P / d \llambda \leq \overline{a} \big\}, \\
     \mathcal{M} &= \big\{ P \in \fPas \colon P \circ X_0^{-1} = \delta_{x_0}, \ (\llambda \otimes P)\text{-a.e. } d B^P / d \llambda = 0, \ \underline{a} \leq d C^P / d \llambda \leq \overline{a} \big\}.
\end{align*}
The idea behind the set \(\mathcal{P}\) of real-world measures is that drift and volatility take flexible values in the random intervals \([\underline{b}, \overline{b}]\) and \([\underline{a}, \overline{a}]\), which capture uncertainty stemming for instance from an estimation procedure. Here, the boundaries of the intervals can depend on the whole history of the paths of the process \(X\) in a predictable manner. 

This setting is included in our framework. For instance, we can model it by taking \(F := [0, 1] \times [0, 1]\) and, for \(((f_1, f_2), t, \omega) \in F \times \of 0, T\gs\), 
\begin{align*}
b ( (f_1, f_2), t, \omega) &:= \underline{b}_t (\omega) + f_1 \cdot ( \overline{b}_t (\omega) - \underline{b}_t (\omega)), \\
a ( (f_1, f_2), t, \omega) &:= \underline{a}_t (\omega) + f_2 \cdot ( \overline{a}_t (\omega) - \underline{a}_t (\omega)). 
\end{align*}
For these choices of \(F, b\) and \(a\), part (i) of Condition \ref{cond: compact, LG, cont} and Condition \ref{cond: convex} are evidently satisfied. 
Furthermore, parts (ii) and (iii) of Condition \ref{cond: compact, LG, cont} transfer directly to the boundary functions (in the sense that (ii) and (iii) are satisfied once \(\underline{b}, \overline{b}, \underline{a}\) and \(\overline{a}\) are continuous and of linear growth). We consider these assumptions to be relatively mild from a practical perspective. 
Finally, let us also comment on Condition~\ref{SA: MPR}, i.e., the existence of a locally bounded MPR. Clearly, in case the coefficients \(\underline{b}, \overline{b}, \underline{a}\) and \(\overline{a}\) are locally bounded, and \(\underline{a}\) is locally bounded away from zero, the function 
\[
\theta ( (f_1, f_2), t, \omega) := \frac{\underline{b}_t (\omega) + f_1 \cdot (\overline{b}_t (\omega) - \underline{b}_t (\omega))}{\underline{a}_t (\omega) + f_2 \cdot (\overline{a}_t (\omega) - \underline{a}_t (\omega))}, \quad ( (f_1, f_2), t, \omega) \in F \times \of 0, T\gs, 
\]
is a robust MPR as described in Condition \ref{SA: MPR}. This setting is uniformly elliptic and therefore, corresponds to a complete market situation. 

Let us also give an example for a related incomplete situation in that Condition~\ref{SA: MPR} is satisfied. Consider the case where \(F := [0, 1]\) and
\begin{align*}
b ( f, t, \omega) := f \cdot \overline{b}_t (\omega), \qquad
a ( f, t, \omega) := f \cdot \overline{a}_t (\omega). 
\end{align*}
In this case, provided we presume that \(\overline{b}\) is locally bounded and \(\overline{a}\) is locally bounded away from zero, the function 
\[
\theta (f, t, \omega) := \frac{\overline{b}_t (\omega)}{\overline{a}_t (\omega)}, \quad (f, t, \omega) \in [0, 1] \times \of 0, T\gs, 
\]
is a robust MPR as described in Condition \ref{cond: MPR bounded}. As \(f = 0\) is possible, the volatility coefficient \(a\) is allowed to vanish.

\subsection{Stochastic delay equations with parameter uncertainty}
In the paper \cite{Fed11}, the author investigates the optimal portfolio strategy for a stochastic delay equation that is used to model a pension fund that provides a minimum guarantee and a surplus that depends on the past performance of the fund itself. We present a framework in the spirit of \cite{Fed11} with uncertain parameters. 
Consider a stochastic process \(Y\) with dynamics
\[
d Y_t = \big[ (r + \sigma^2 \lambda ) Y_t - \gamma (Y_t - Y_{(t - \tau) \vee 0}) \big] dt + \sigma d W_t.
\]
Here, the constants \(r, \sigma, \lambda\) and the function \(\gamma \colon \mathbb{R} \to \bR_+\), that we presume to be of linear growth, are model parameters and \(W\) is a one-dimensional standard Brownian motion. 
The term 
\(
\gamma (Y_t - Y_{(t - \tau) \vee 0})
\)
is related to a surplus part of benefits of a fund and \(\tau \in [0, T]\) represents the time people remain in the fund. The dynamics of \(Y\) follow a so-called {\em delay equation} that has a non-Markovian structure due to the presence of the surplus term. 

In the following, we shortly explain how uncertainty can be introduced to such a framework. 
Let \(\underline{r}, \overline{r}, \underline{\lambda}, \overline{\lambda}, \underline{\sigma}^2\) and \(\overline{\sigma}^2\) be constants such that 
\[
\underline{r} \le \overline{r}, \quad \underline{\lambda} \leq \overline{\lambda}, \quad 
0 < \underline{\sigma}^2 \leq \overline{\sigma}^2,
\]
and define
\[
F := [\underline{r}, \overline{r}] \times [\underline{\lambda}, \overline{\lambda}] \times [\underline{\sigma}^2, \overline{\sigma}^2].
\]
Clearly, \(F\) is compact and part (i) from Condition \ref{cond: compact, LG, cont} holds.
Next, we introduce the robust coefficients \(b\) and \(a\) as
\begin{align*}
b (f, t, \omega) := (f_1 + f_3 f_2) \omega (t) - \gamma (\omega (t) - \omega ((t - \tau) \vee 0)), \quad
a (f, t, \omega) := f_3,
\end{align*}
for \(( (f_1, f_2, f_3), t, \omega) \in F \times \of 0, T\gs\).
The functions \(b\) and \(a\) satisfy (ii) and (iii) from Condition \ref{cond: compact, LG, cont}. Furthermore, a short computation shows that also Condition \ref{cond: convex} holds.

Since \(\underline{\sigma}^2 > 0\), the function 
\[
\theta ( (f_1, f_2, f_3), t, \omega) := \frac{ (f_1 + f_3 f_2) \omega (t) - \gamma (\omega (t) - \omega ((t - \tau) \vee 0))}{f_3}, \quad ((f_1, f_2, f_3), t, \omega) \in F \times \of 0, T\gs, 
\]
is a MPR that satisfies Condition \ref{SA: MPR}. Finally, we notice that Condition \ref{cond: unif ellipticity vola} holds. Hence, our main Theorems \ref{thm: main positive power exponential} and \ref{thm: main negative power log} apply in this setting.

\section{Proofs} \label{sec: proofs}
In this section we present the proofs of our main results. The structure is chronological with the appearance of the results in the previous sections.

\subsection{Superhedging duality: proof of Theorem \ref{thm: superhedging duality}} \label{sec: proof superhedging duality}
We start with some auxiliary preparations and then finalize the proof.
For \(\omega, \omega' \in \Omega\) and \(t \in [0, T]\), we define the concatenation
\[
\omega \otimes_t \omega' :=  \omega \1_{[ 0, t)} + (\omega (t) + \omega' - \omega' (t)) \1_{[t, T]}.
\]
Furthermore, for a probability measure \(Q \in \mathfrak{P}(\Omega)\) and a transition kernel \(\omega \mapsto Q^*_\omega\), we set
\[
(Q\otimes_\tau Q^*) (A) := \iint \1_A (\omega \otimes_{\tau(\omega)} \omega') Q^*_\omega (d \omega') Q(d \omega), \quad A \in \mathcal{F}.
\]

\begin{definition}
A family \( \{\cR(t, \omega) \colon (t,\omega) \in \of 0, T\gs\} \subset \mathfrak{P}(\Omega)\) is said to have the {\em Property (A)} if 
\begin{enumerate}
   \item[\textup{(i)}]  the graph
\(
\on{gr} \cR = \big\{ (t, \omega, Q) \in \of 0, T \gs \hspace{0.05cm} \times\hspace{0.05cm} \mathfrak{P}(\Omega) \colon Q \in \cR(t, \omega) \big\}
\)
is analytic;
\item[\textup{(ii)}]
for any \((t, \alpha) \in \of 0, T\gs\), any stopping time \(\tau\) with \(t \leq \tau \leq T\), and any \(Q \in \cR(t, \alpha)\) there exists a family \(\{Q (\cdot | \mathcal{F}_\tau) (\omega) \colon\) \(\omega \in \Omega \}\) of regular \(Q\)-conditional probabilities given \(\mathcal{F}_\tau\) such that \(Q\)-a.s. \(Q (\cdot | \mathcal{F}_\tau) \in \cR(\tau, X)\);
\item[\textup{(iii)}]
for any \((t, \alpha) \in \of 0, T\gs\), any stopping time \(\tau\) with \(t \leq \tau \leq T\), any \(Q \in \cR(t, \alpha)\) and any \(\mathcal{F}_\tau\)-measurable map \(\Omega \ni \omega \mapsto Q^*_\omega \in \mathfrak{P}(\Omega)\) the following implication holds:
\[
Q\text{-a.s. } Q^* \in \cR (\tau, X)\quad \Longrightarrow \quad Q \otimes_\tau Q^* \in \cR(t, \alpha).
\]
\end{enumerate}
\end{definition}

\begin{definition}
	We say that a family \(\{ \cR (t, \omega) \colon (t, \omega) \in \of 0, T\gs\} \subset \fP (\Omega)\) is \emph{adapted} if \[\cR (t, \omega) = \cR(t, \omega(\cdot \wedge t))\] for all \((t, \omega) \in \of 0, T \gs\).
\end{definition}

\begin{lemma} \label{lem: stability equivalent measures}
    Let \(\{\cR (t, \omega) \colon (t, \omega) \in \of 0, T\gs\}\) be an adapted family that satisfies Condition (A) and that has a Borel measurable graph. Then, the family \(\{\fP_a(\cR (t, \omega)) \colon (t, \omega) \in \of 0, T\gs\}\), where
     \[
    \fP_a(\cR (t, \omega)) = \{P \in \fP(\Omega) \colon \exists R \in \cR(t,\omega) \text{ with } P \ll R \}, \quad (t, \omega) \in \of 0, T\gs,
    \]
    is adapted and satisfies Condition (A).
\end{lemma}
\begin{proof}
    \emph{Step 1.} We start by showing that the correspondence \(\fP_a (\cR)\) has an analytic graph. As \(\cF\) is countably generated, \cite[Theorem V.58, p. 52]{DM} (and the subsequent remarks) grants the existence of a Borel function 
    \( D \colon \Omega \times \fP(\Omega) \times \fP(\Omega) \to \bR_+\) such that \(D(\cdot, P, R)\) is a version of the Radon–Nikodym derivative of the absolutely continuous part of \(P\) with respect to \(R\) on \(\cF\).
    Let
    \begin{align} \label{eq: phi}
    \fP(\Omega) \times \fP(\Omega) \ni (P, R) \mapsto \phi (P, R) := E^R\big[ D(\cdot, P, R) \big] \in [0, 1].
    \end{align}
    Notice that \(\phi\) is Borel by \cite[Theorem~8.10.61]{bogachev}.
    Let \(\pi \colon \of 0, T \gs \times \fP (\Omega) \times \fP (\Omega) \to \of 0, T \gs \times \fP (\Omega)\) be the projection to the first three coordinates. As \(\cR\) is assumed to have a Borel measurable graph, we conclude from the identity
    \begin{align*}
    \on{gr} \fP_a (\cR) 
    &= \pi \big( \big\{ (t, \omega, P, R) \colon (t, \omega, R) \in \on{gr} \cC, \phi (P, R) = 1 \big\} \big), 
    \end{align*}
    that \(\fP_a (\cR)\) has an analytic graph.

    \emph{Step 2.} Next, we show part (ii) from Condition (A). Let \((t, \alpha) \in \of 0, T\gs\) and take a stopping time \(t \leq \tau \leq T\) and a measure \(Q \in \fP_a(\cR (t, \alpha))\). By definition, there exists a measure \(P \in \cR (t, \alpha)\) such that \(Q \ll P\). As \(\cR\) satisfies part (ii) from Condition (A), \(P\)-a.s. \(P (\cdot | \cF_\tau) \in \cR (\tau, X)\). To conclude that \(Q\)-a.s. \(Q (\cdot | \cF_\tau) \in \fP_a (\cR(\tau, X))\), it suffices to show that \(Q\)-a.s. \(Q (\cdot | \cF_\tau) \ll P(\cdot | \cF_\tau)\). With the notation \(Z := d Q / d P\), the generalized Bayes theorem (\cite[Theorem 6, p. 274]{shiryaevProb}) yields that, for all \(A \in \cF\), \(Q\)-a.s.
    \[
Q (A | \cF_\tau) = \frac{ E^P [ \1_A Z | \cF_\tau]}{E^P [ Z | \cF_\tau]}. 
    \]
    Due to the fact that \(\cF\) is countably generated, this formula holds \(Q\)-a.s. for all \(A \in \cF\). Hence, \(Q\)-a.s. \(Q (\cdot | \cF_\tau) \ll P(\cdot | \cF_\tau)\), which proves that \(Q\)-a.s. \(Q ( \cdot | \cF_\tau) \in \fP_a (\cR(\tau, X))\).

    {\em Step 3.} It remains to prove part (iii) from Condition (A). Take \((t, \alpha) \in \of 0, T\gs\), let \(\tau\) be a stopping time with \(t \leq \tau \leq T\), fix a measure \(Q \in \fP_a (\cR(t, \alpha))\) and an \(\mathcal{F}_\tau\)-measurable map \(\Omega \ni \omega \mapsto Q^*_\omega \in \mathfrak{P}(\Omega)\) such that \(Q\)-a.s. \(Q^* \in \fP_a (\cR( \tau, X))\). 
    Using that \(\cR\) is adapted, \(\on{gr} \cR\) is Borel and that \(\omega \mapsto Q^*_\omega\) is \(\cF_\tau\)-measurable, we obtain that
    \begin{align*}
    	\mathcal{Z} &:= \big\{ (\omega, P) \in \Omega \times \fP (\Omega) \colon P \in \cR (\tau (\omega), \omega), Q^*_\omega \ll P \big\}
    	\\&= \big\{ (\omega, P) \in \Omega \times \fP (\Omega) \colon (\tau (\omega), \omega (\cdot \wedge \tau(\omega)), P) \in \on{gr} \cR, \phi ( Q^*_{\omega}, P) = 1\big\} 
    	\in \cF_\tau \otimes \mathcal{B}(\fP(\Omega)).
    \end{align*}
Denoting the projection to the first coordinate by \(\pi_1 \colon \Omega \times \fP(\Omega) \to \Omega\), it follows that the set 
\[
\pi_1 (\mathcal{Z}) = \big\{ \omega \in \Omega \colon \exists P \in \cR (\tau(\omega), \omega), Q_\omega^* \ll P \big\} 
\]
is analytic and consequently, an element of \(\cF^*\). Using Galmarino's test, we also see that 
\[
\pi_1 (\mathcal{Z}) = \big\{ \omega \in \Omega \colon \omega (\cdot \wedge \tau (\omega)) \in \pi_1 (\mathcal{Z}) \big\}.
\]
Hence, it follows from the universally measurable version of Galmarino's test (\cite[Lemma~2.5]{NVH}) that \(\pi_1 (\mathcal{Z})\in \cF^*_\tau\). 
As \(Q \in \fP_a (\cR (t, \alpha))\), there exists a measure \(P \in \cR (t, \alpha)\) such that \(Q \ll P\).
We define a nonempty-valued correspondence \(\gamma \colon \Omega \twoheadrightarrow \fP(\Omega)\) by 
\[
\gamma (\omega) := \begin{cases} \big\{ R \in \fP (\Omega) \colon R \in \cR (\tau (\omega), \omega), \ Q^*_\omega \ll R \big\}, & \omega \in \pi_1 (\mathcal{Z}), \\
	\big\{ P (\cdot | \cF_\tau) (\omega) \big\}, & \omega \not \in \pi_1 (\mathcal{Z}). \end{cases}
\]
Notice that
\[
\on{gr} \gamma = \big[ \mathcal{Z} \cap (\pi_1 (\mathcal{Z}) \times \fP(\Omega)) \big] \cup \big[ \on{gr} P (\cdot | \cF_\tau) \cap (\Omega \backslash \pi_1 (\mathcal{Z}) \times \fP(\Omega)) \big].
\]
As \(\omega \mapsto P(\cdot | \cF_\tau) (\omega)\) is \(\cF_\tau\)-measurable and the image space is Polish,\footnote{cf. \cite[Exercise 3.10.53]{bogachev}} we observe that \(\on{gr} P (\cdot | \cF_\tau) \in \cF_\tau \otimes \mathcal{B}(\fP(\Omega))\) and therefore, that \(\on{gr} \gamma \in \cF^*_\tau \otimes \mathcal{B}(\fP(\Omega))\).
By Aumann's theorem (\cite[Theorem 5.2]{himmelberg}), there exists an \(\cF^*_\tau\)-measurable function \(\Omega \ni \omega \mapsto \overline{P}_\omega \in \fP(\Omega)\) such that \(P\)-a.s. \(\overline{P} \in \gamma\). It is well-known (\cite[Lemma~1.27]{Kallenberg}) that \(\overline{P}\) coincides \(P\)-a.s. with an \(\cF_\tau\)-measurable function \(\Omega \ni \omega \mapsto P^*_\omega \in \fP(\Omega)\). As \(P \in \cR (t, \alpha)\) and because \(\cR\) satisfies part (ii) from Property (A), we have \(P\)-a.s. \(P (\cdot | \cF_\tau) \in \cR(\tau, X)\). Consequently, \(P\)-a.s. \(P^* \in \cR(\tau, X)\). Further, as \(Q \ll P\) and because \(Q (\pi_1 (\mathcal{Z})) = 1\), \(Q\)-a.s. \(Q^* \ll P^*\).
We are in the position to complete the proof. 
Using that \(\cR\) satisfies part (iii) from Condition (A), we have
    \(
    P \otimes_\tau P^* \in \cR (t, \alpha).
    \)
    Hence, it suffices to show that \(Q \otimes_\tau Q^* \ll P \otimes_\tau P^*\).
    Let \(A \in \cF\) such that \((P \otimes_\tau P^*) (A) = 0\). Then, by definition of \(P \otimes_\tau P^*\), we have \(P^*_\omega (\{ \omega' \colon \omega \otimes_{\tau (\omega)} \omega' \in A\}) = 0\) for \(P\)-a.a., and because \(Q \ll P\) also \(Q\)-a.a., \(\omega \in \Omega\). Since \(Q\)-a.s. \(Q^* \ll P^*\), we get that \(Q^*_\omega (\{ \omega' \colon \omega \otimes_{\tau (\omega)} \omega' \in A\}) = 0\) for \(Q\)-a.a. \(\omega \in \Omega\), which implies that \((Q \otimes_\tau Q^*) (A) = 0\).
    We conclude that \(Q \otimes_\tau Q^* \in \fP_a (\cR(t, \alpha))\). The proof is complete.
\end{proof}

We call an \(\bR^d\)-valued continuous process \(Y = (Y_t)_{t \in [0, T]}\) a (continuous) \emph{semimartingale after a time \(t^* \in \mathbb{R}_+\)} if the process \(Y_{\cdot + t^*} = (Y_{t + t^*})_{t \in [0, T - t^*]}\) is a \(d\)-dimensional semimartingale for its natural right-continuous filtration.
The law of a semimartingale after \(t^*\) is said to be a \emph{semimartingale law after \(t^*\)} and the set of them is denoted by \(\fPs (t^*)\).
Notice also that \(P \in \fPs(t^*)\) if and only if the coordinate process is a semimartingale after \(t^*\), see \cite[Lemma 6.4]{CN22}.
For \(P \in \fPs (t^*)\) we denote the semimartingale characteristics of the shifted coordinate process \(X_{\cdot + t^*}\) by \((B^P_{\cdot + t^*}, C^P_{\cdot + t^*})\). 
Moreover, we set 
\[
\fPas (t^*) := \big\{ P \in \fPs (t^*) \colon P\text{-a.s. } (B^P_{\cdot + t^*}, C^P_{\cdot + t^*}) \ll \llambda \big\}, \quad \fPas := \fPas (0),
\]
where \(\llambda\) denotes the Lebesgue measure. 
For $(t,\omega) \in \of 0, T\gs$, we define 
\begin{equation*}\label{eq: P(t,w)}
   \cP(t,\omega) := \Big\{ P \in \mathfrak{P}_{\text{sem}}^{\text{ac}}(t)\colon P(X^t = \omega^t) = 1, (\llambda \otimes P)\text{-a.e. } (dB^P_{\cdot + t} /d\llambda, dC^P_{\cdot + t}/d\llambda) \in \Theta (\cdot + t, X) \Big\}, 
\end{equation*} 
and
\begin{equation*}\label{eq: M(t,w)}
   \cM(t,\omega) := \Big\{ P \in \mathfrak{P}_{\text{sem}}^{\text{ac}}(t)\colon P(X^t = \omega^t) = 1, (\llambda \otimes P)\text{-a.e. } (dB^P_{\cdot + t} /d\llambda, dC^P_{\cdot + t}/d\llambda) \in \tilde{\Theta} (\cdot + t, X) \Big\}, 
\end{equation*} 
where we use the standard notation \(X^t := X_{\cdot \wedge t}\).

\begin{corollary} \label{cor: Condition (A)}
The family \(\{\cM (t, \omega) \cap \fP_a(\cP (t, \omega)) \colon (t, \omega) \in \of 0, T\gs\}\) satisfies Condition (A) and is adapted.
\end{corollary}
\begin{proof}
Notice that \(\{\cP (t, \omega) \colon (t, \omega) \in \of 0, T\gs\}\) and \(\{\cM (t, \omega) \colon (t, \omega) \in \of 0, T\gs\}\) are adapted by construction. Further, for those two families, 
Condition (A) has been verified in \cite[Lemmata 6.6, 6.12, 6.17]{CN22}, including Borel measurability of their graphs. Hence, thanks to Lemma \ref{lem: stability equivalent measures},
the family \( \{ \fP_a(\cP (t, \omega)) \colon (t, \omega) \in \of 0, T\gs\}\) satisfies Condition (A) and is adapted. Therefore, the intersection \(\{\cM (t, \omega) \cap \fP_a(\cP (t, \omega)) \colon (t, \omega) \in \of 0, T\gs\}\) satisfies Condition (A) and is adapted. 

\end{proof}

The following lemma provides a minor extension of the dynamic programming principle as given by
\cite[Theorem 2.1]{ElKa15}. This is very much in the spirit of \cite[Theorem 2.3]{NVH}, whose proof we follow.
\begin{lemma} \label{lem: supermartingale}
    Let \(\{\cR (t, \omega) \colon (t, \omega) \in \of 0, T\gs\}\) be an adapted family that satisfies Condition (A) and let \(f \colon \Omega \to \overline{\bR}\) be an upper semianalytic function.
    Define, for \((t, \omega) \in \of 0, T\gs\),
    \[
    \cE_t(f)(\omega) := \sup_{P \in \cR(t,\omega)} E^P[f].
    \]
    Let \(s,t \in [0,T ]\) with \(s \leq t\).
    Then, for fixed \(\omega \in \Omega\) and 
    \(P \in \cR(0,\omega)\), we have
    \[
    \cE_s(f) = \esssup^P \big\{
    E^R[ \cE_t(f) | \cF_s ] \colon R \in \cR(0,\omega) \text{ with } R = P \text{ on } \cF_s \big\} \quad \text{\(P\)-a.s.}
    \]
\end{lemma}
\begin{proof}
    Fix \(\omega \in \Omega\) and \(P \in \cR(0,\omega)\).
    We start by showing, for \(s \in [0,T]\) and every upper semianalytic function \(f \colon \Omega \to \overline{\bR}\),
    \begin{equation} \label{eq: pf supermartingale}
        \cE_s(f) = \esssup^P \big\{
    E^R[ f | \cF_s ] \colon R \in \cR(0,\omega) \text{ with } R = P \text{ on } \cF_s \big\} \quad \text{\(P\)-a.s.}
    \end{equation}
    Let \(R \in \cR(0,\omega)\)  with  \(R = P\) on \(\cF_s \).
    By Condition (A), there exists a family \(\{R (\cdot | \mathcal{F}_s) (\alpha) \colon\) \(\alpha \in \Omega \}\) of regular \(R\)-conditional probabilities given \(\mathcal{F}_s\) such that \(R\)-a.s. \(R (\cdot | \mathcal{F}_s) \in \cR(s, X)\). Hence,
    \begin{equation*} \label{eq: pf supermartingale 2}
        \cE_s(f) \geq E^R[ f | \cF_s ] \quad \text{ \(R\)-a.s.}
    \end{equation*}
    Notice that both sides in \eqref{eq: pf supermartingale 2} are \(\cF_s^*\)-measurable by \cite[Theorem 2.1]{ElKa15} and a suitable version of Galmarino's test, see \cite[Lemma 2.5]{NVH}. As \(R = P\) on \(\cF_s\) we also have \(R = P\) on \(\cF^*_s\) and we conclude that 
    \[
    \cE_s(f) \geq \esssup^P \big\{
    E^R[ f | \cF_s ] \colon R \in \cR(0,\omega) \text{ with } R = P \text{ on } \cF_s \big\} \quad \text{\(P\)-a.s.}
    \]
    To show the converse inequality let \(\epsilon > 0\). As in the proof of \cite[Theorem 2.1]{ElKa15} there exists an \(\cF_s\)-measurable kernel \(\Omega \ni \omega \mapsto Q^*_\omega \in \mathfrak{P}(\Omega)\) such that 
    \(P\)-a.s. \(Q^* \in \cR (t, X)\) and
    \[
    E^{Q^*}[f] \geq \cE_s(f) \1_{\{\cE_s(f) < \infty\}} + \frac{1}{\epsilon}\1_{ \{\cE_s(f) = \infty\}}.
    \]    
    By Condition (A), the measure \(P \otimes_s Q^*\) is contained in \(\cR(0,\omega)\) and coincides with \(P\) on \(\cF_s\). Further, it holds that
    \[
    E^{P \otimes_s Q^*}[ f \mid \cF_s] = E^{Q^*}[f] \geq (\cE_s(f) - \epsilon) \wedge \frac{1}{\epsilon} \quad \text{\(P\)-a.s.},
    \]
    and, as \(\epsilon > 0\) was arbitrary, we conclude that \eqref{eq: pf supermartingale} holds.
    Next, it follows from \cite[Theorem 2.1]{ElKa15} and \eqref{eq: pf supermartingale} that 
    \[
    \cE_s(f) = \cE_s(\cE_t(f)) = \esssup^P \big\{
    E^R[ \cE_t(f) | \cF_s ] \colon R \in \cR(0,\omega) \text{ with } R = P \text{ on } \cF_s \big\} \quad \text{\(P\)-a.s.}
    \]
    The proof is complete.
\end{proof}

We are in the position to give a proof for Theorem \ref{thm: superhedging duality}.
We follow closely the lines of \cite[proof of Theorem 3.2]{nutz_superhedging} and adapt it to our setting.
\begin{proof}[Proof of Theorem \ref{thm: superhedging duality}]
By Corollary \ref{cor: Condition (A)}, the family \(\{\cR(t,\omega) \colon (t, \omega) \in \of 0, T\gs\}\), where
\[
\cR(t,\omega) := \cM (t, \omega) \cap \fP_a(\cP (t, \omega)), \quad (t, \omega) \in \of 0, T\gs,
\]
satisfies Condition (A) and is adapted.
Hence, it follows from \cite[Theorem 2.1]{ElKa15} that for every upper semianalytic \(f \colon \Omega \to \overline{\bR}\) and \(t \in [0,T]\) the function
\[
\cE_t(f)(\omega) := \sup_{Q \in \cR(t,\omega)} E^Q[f], \quad \omega \in \Omega,
\]
is \(\cF^*_t\)-measurable and satisfies \(\cE_s(\cE_t(f)) = \cE_s(f)\) for \(s, t \in [0,T]\) with \(s \leq t\). Let \(f \colon \Omega \to \bR_+\) be as in the statement of the theorem. Notice that \(\fM_a(\cP) = \cR(0, \bx_0)\), where \(\bx_0 \in \Omega\) is the constant path \(\bx_0 (s) = x_0\) for all \(s \in [0,T]\). Hence, we conclude that
\[
\sup_{Q \in \fM_a(\cP)} E^Q[\cE_t(f)] = \pi < \infty,
\]
where \(\pi := \sup_{Q \in \fM_a(\cP)} E^Q \big[f \big]\).
This, together with Lemma \ref{lem: supermartingale}, implies that the process \(t \mapsto \cE_t(f)\)
is a \(Q\)-\(\F^*\)-supermartingale for every \(Q \in \fM_a(\cP)\).
As in the proof of \cite[Theorem 3.2]{nutz_superhedging} we can now construct a \cadlag process \(Y\) that is a \(Q\)-\(\G\)-supermartingale for every \(Q \in \fM_a(\cP)\) and satisfies
\[
Y_0 \leq \pi \quad \text{ and } \quad Y_T = f \quad \text{\(Q\)-a.s. for all \(Q \in \fM_a(\cP)\)}.
\]
As \(\fM_a(\cP)\) is a nonempty and saturated (in the sense of \cite{nutz_superhedging}) set of local martingale measures for \(X\),
the robust optional decomposition theorem \cite[Theorem~2.4]{nutz_superhedging}
grants the existence of a \(\G^{\fM_a(\cP)}\)-predictable process \(H\) such that \(\int_0^\cdot H_s d X_s\) is a \(Q\)-supermartingale for every \(Q \in \fM_a(\cP)\), and
\[
\pi + \int_0^T H_s d X_s \geq f \quad Q\text{-a.s. for all } Q \in \fM_a(\cP).
\]
In particular, as \(f\) is non-negative, this implies
\[
\int_0^t H_s d X_s \geq E^Q\Big[\int_0^T H_s d X_s \mid \cG^{\fM_a(\cP)} \Big] \geq -\pi, \quad \text{ for all } t \in [0, T],\ Q\text{-a.s. for all } Q \in \fM_a(\cP).
\]
This completes the proof.
\end{proof}

\subsection{Separating duality for nonlinear semimartingales: proof of Theorem \ref{thm: dualities}} \label{sec: proof separating duality}

\begin{lemma} \label{lem: integrands}
    Suppose that \NFLVR holds.
    Let \(H\) be a \(\G^\cP\)-predictable process. Then, \(H \in L (X, P)\) for every \(P\in \cP\) if and only if \(H \in L (X, Q)\) for every \(Q\in \fM_a(\cP)\).
     In particular, \( \cH^\cP = \cH^{\fM_a(\cP)}\).
\end{lemma}
\begin{proof}
Recall that \NFLVR implies that \(\cP \sim \fM_a(\cP)\). Hence, it suffices to establish the first part of the statement.
In this regard, suppose \(H \in L (X, P)\) for every \(P\in \cP\),
and let \(Q \in \fM_a(\cP)\). As \(\fM_a(\cP) \subset \fP_a(\cP)\), there exists a measure \(P \in \cP\) with \(Q \ll P\).
Hence, \cite[Lemma V.2]{memin} implies that \(H \in L (X, Q)\).
Conversely, let \(H \in L (X, Q)\) for every \(Q \in \fM_a(\cP)\),
and let \(P \in \cP\). By \NFLVR, there exists a measure \(Q \in \fM_a(\cP)\) with \(P \ll Q\).
Hence, \cite[Lemma V.2]{memin} implies that \(H \in L (X, P)\).
\end{proof}

\begin{proposition}[First Duality]  \label{prop: first duality}
    \(\fM_a(\cP) = \cD = \big\{ Q \in \mathfrak{P}_a(\cP) \colon E^Q \big[g \big] \leq 1 \text{ for all } g \in \cC \cap C_b  (\Omega; \mathbb{R})\big \} =: \cU \)
\end{proposition}
\begin{proof}
To see \( \fM_a(\cP) \subset \cD\), let \(g \in \cC\) and \(Q \in \fM_a(\cP)\). By definition of \(\fM_a (\cP)\), there exists a measure \(P \in \cP\) with \(Q \ll P\). Moreover, by the definition of \(\cC\), there exists a process \(H \in \cH^\cP \subset L(X,P)\) such that \(P\)-a.s.
\begin{equation} \label{eq: pf first duality}
    g \leq 1 + \int_0^T H_s d X_s.
\end{equation}
We deduce from \cite[Lemma V.2]{memin} that \eqref{eq: pf first duality} holds 
\(Q\)-a.s. as well.
As \(1 + \int_0^\cdot H_s d X_s\) is a non-negative local \(Q\)-martingale, it is also a \(Q\)-supermartingale, which shows that \( E^Q[g] \leq 1\). Thus, \(Q \in \cD\).
Notice that \(\cD \subset \mathcal{U}\) by definition.
It remains to show \(\mathcal{U} \subset \fM_a(\cP)\).
To this end, we define
\[
\Gamma := \Big\{ g \in C_b (\Omega; \bR)  \colon \exists H \in \cH^\cP \text{ such that } g \leq \int_0^T H_s d X_s \Big\}. 
\]
Let \(Q \in \mathcal{U}\) and take a function \(g \in \Gamma\). 
As \(g\) is bounded, there exists a constant \(c > 0\) such that \(g + c \geq 0\). Moreover, as \(Q \in \mathcal{U} \subset \fP_a(\cP)\), there exists a process \(H \in \cH^\cP\) such that \(Q\)-a.s.
\(
1 + g/c \leq 1 +  \int_0^T H_s / c d X_s.
\)
Hence, because \(H/c \in \cH^\cP\), we have \(1 + g/c \in \cC \cap C_b (\Omega; \bR)\) and the definition of \(\mathcal{U}\) yields that 
\(
1 +E^Q [ g ] / c \leq 1.
\)
This implies \(E^Q[g] \leq 0\). 
Thanks to \cite[Lemma 5.6]{kupper}, we can conclude that \(Q\) is a local martingale measure. This yields \(Q \in \fM_a(\cP)\).
\end{proof}

\begin{proposition} [Second Duality] \label{prop: second duality}
Suppose that \NFLVR holds.
Then, 
\[
 \cC \cap C_b (\Omega; \bR) = \big\{ g \in C_b^+(\Omega; \bR) \colon E^Q \big[g\big] \leq 1 \text{ for all } Q \in \cD\big\}.  
 \]
\end{proposition}
\begin{proof}
    Since, by definition, \( \cC \cap C_b (\Omega; \bR) \subset \{ g \in C_b^+ (\Omega; \bR) \colon E^Q[g] \leq 1 \text{ for all } Q \in \cD\}\), it suffices to prove the converse inclusion \( \cC \cap C_b (\Omega; \bR) \supset \{ g \in C_b^+ (\Omega; \bR) \colon E^Q[g] \leq 1 \text{ for all } Q \in \cD\}\). 
    Let \(g \in C_b^+ (\Omega; \bR) \) be such that \(E^Q[g] \leq 1\) for all \(Q \in \cD\).
    As \(g\) is bounded and continuous, the superhedging duality given by Theorem \ref{thm: superhedging duality}, together with the equality \(\cD = \fM_a(\cP)\) from Proposition \ref{prop: first duality}, grants the existence of a \(\G^{\fM_a(\cP)}\)-predictable process \(H\) such that \(H \in \cH^{\fM_a(\cP)}\), and
    \[
    1 + \int_0^T H_s d X_s \geq g \quad Q\text{-a.s. for all } Q \in \fM_a(\cP).
    \]
    As \(H \in \cH^{\cP}\) by Lemma \ref{lem: integrands}, and \(\cP\sim \fM_a (\cP)\) by the hypothesis that \NFLVR holds, we conclude that \(g \in \cC \cap C_b\).
    The proof is complete.
\end{proof}

Now, Theorem \ref{thm: dualities} follows directly from the previous two propositions. 
\begin{proof}[Proof of Theorem \ref{thm: dualities}] The duality \eqref{eq: first duality} follows from Proposition \ref{prop: first duality}, and the duality \eqref{eq: second duality} follows from Proposition~\ref{prop: second duality}.
\end{proof}

\subsection{\(\cP\) and \(\cM\) are convex and compact: proof of Theorem \ref{thm: compactness}} \label{sec: compactness}

It suffices to prove the claims for \(\cP\), as \(\cM\) is the special case with \(b \equiv 0\).
The convexity follows from \cite[Lemma~III.3.38, Theorem~III.3.40]{JS}, see the proof of \cite[Lemma 5.8]{CN22b} for more details.
Next, we prove compactness. First, we show that \(\cP\) is closed, which follows along the lines of the proof of \cite[Proposition 3.8]{CN22b}, and then we explain that \(\cP\) is also relatively compact.

Let \((P^n)_{n \in \mathbb{N}} \subset \cP\) be such that \(P^n \to P\) weakly. We have to show that \(P \in \cP\), i.e., we have to prove that \(P\in \fPas\) with differential characteristics in \(\Theta\).
For each \(n \in \mathbb{N}\), denote the \(P^n\)-characteristics of \(X\) by \((B^n, C^n)\).

Before we start the main part of this proof, we need a last bit of notation. Let 
\[\Omega' := \Omega \times \Omega \times C([0, T]; \bR^{d \times d})\] and denote the coordinate process on \(\Omega'\) by \(Y = (Y^{(1)}, Y^{(2)}, Y^{(3)})\). Further, set \(\cF' := \sigma (Y_s, s \in [0, T])\) and let \(\F' = (\cF'_s)_{s \in [0, T]}\) be the right-continuous filtration generated by \(Y\).

\emph{Step 1.} We start by showing that \(\{P^n \circ (X, B^n, C^n)^{-1} \colon n \in \mathbb{N}\}\) is tight. 
Since \(P^n \to P\), it suffices to prove tightness of \(\{P^n \circ (B^n, C^n)^{-1} \colon n \in \mathbb{N}\}\). We use Aldous' tightness criterion (\cite[Theorem~VI.4.5]{JS}), i.e., we show the following two conditions:
\begin{enumerate}
    \item[(a)]
    for every \(\varepsilon > 0\), there exists a \(K > 0\) such that
    \[
    \sup_{n \in \mathbb{N}} P^n \Big( \sup_{s \in [0, T]} \|B^n_s\| + \sup_{s \in [0, T]} \|C^n_s\| \geq K \Big) \leq \varepsilon;
    \]
    \item[(b)] 
    for every \(\varepsilon > 0\),
    \[
    \lim_{\theta \searrow 0} \limsup_{n \to \infty} \sup \big\{P^n (\|B^n_L - B^n_S\| + \|C^n_L - C^n_S\| \geq \varepsilon) \big\} = 0, 
    \]
    where the \(\sup\) is taken over all stopping times \(S, L \leq T\) such that \(S \leq L \leq S + \theta\).
\end{enumerate}
By the linear growth assumptions on \(b\) and \(a\) from Condition \ref{cond: compact, LG, cont}, a standard Gronwall argument (see, e.g., \cite[Problem~5.3.15]{KaraShre}) shows that 
\[
\sup_{P \in \cP} E^P \Big[ \sup_{s \in [0, T]} \|X_s\|^2 \Big] < \infty.
\]
Thus, we get
\begin{align} \label{eq: second moment bound}
\sup_{n \in \mathbb{N}} E^{P^n} \Big[ \sup_{s \in [0, T]} \|X_s\|^2 \Big] 
\leq \sup_{P \in \cP} E^P \Big[ \sup_{s \in [0, T]} \|X_s\|^2 \Big]
< \infty.
\end{align}
Using the linear growth assumption once again, we obtain that \(P^n\)-a.s.
\[
\sup_{s \in [0, T]} \|B^n_s\| + \sup_{s \in [0, T]} \|C^n_s\| \leq \C \Big( 1 + \sup_{s \in [0, T]} \|X_s\|^2 \Big),
\]
where the constant \(\C>0\) is independent of \(n\).
By virtue of \eqref{eq: second moment bound}, this bound yields~(a). For (b), take two stopping times \(S, L \leq T\) such that \(S \leq L \leq S + \theta\) for some \(\theta > 0\). Then, using again the linear growth assumptions, we get \(P^n\)-a.s.
\[
\|B^n_L - B^n_S\| + \|C^n_L - C^n_S\| \leq \C (L - S) \Big( 1 + \sup_{s \in [0, T]} \|X_s\|^2 \Big) \leq \C \theta \Big( 1 + \sup_{s \in [0, T]} \|X_s\|^2 \Big),
\]
which yields (b) by virtue of \eqref{eq: second moment bound}. We conclude that the family \(\{P^n \circ (X, B^n, C^n)^{-1} \colon n \in \mathbb{N}\}\) is tight. 
Up to passing to a subsequence, from now on we assume that \(P^n \circ (X, B^n, C^n)^{-1} \to Q\) weakly.

\emph{Step 2.} Next, we show that \(Y^{(2)}\) and \(Y^{(3)}\) are \(Q\)-a.s. absolutely continuous.
For \(M > 0\) and \(\omega \in \Omega\), define 
\[
\tau_M (\omega) := \inf \{t \in [0, T] \colon \|\omega(t)\| \geq M\} \wedge T.
\]
Furthermore, for \(\omega = (\omega^{(1)}, \omega^{(2)}) \in \Omega \times \Omega\), we set 
\[
\zeta_M (\omega) := \sup \Big\{ \frac{\|\omega^{(2)} (t \wedge \tau_M(\omega)) - \omega^{(2)}(s \wedge \tau_M(\omega))\|}{t - s} \colon 0 \leq s < t \leq T\Big\}.
\]
Similar to the proof of \cite[Lemma 3.6]{CN22b}, we obtain the existence of a dense set \(D \subset \bR_+\) such that for every \(M \in D\) the map \(\zeta_M\) is \(Q \circ (Y^{(1)}, Y^{(2)})^{-1}\)-a.s. lower semicontinuous. By the linear growth conditions and the definition of \(\tau_M\), for every \(M \in D\) there exists a constant \(\C = \C (M) > 0\) such that \(P^n(\zeta_M (X, B^n) \leq \C) = 1\) for all \(n \in \mathbb{N}\). As \(\zeta_M\) is \(Q \circ (Y^{(1)}, Y^{(2)})^{-1}\)-a.s. lower semicontinuous, \cite[Example 17, p. 73]{pollard} yields that 
\[
0 = \liminf_{n \to \infty} P^n (\zeta_M (X, B^n) > \C) \geq Q (\zeta_M (Y^{(1)}, Y^{(2)}) > \C). 
\]
Further, since \(D\) is dense in \(\bR_+\), we obtain that \(Y^{(2)}\) is \(Q\)-a.s. Lipschitz continuous, i.e., in particular absolutely continuous. Similarly, we get that \(Y^{(3)}\) is \(Q\)-a.s. Lipschitz and hence, absolutely continuous. 

\emph{Step 3.} 
Define the map \(
\Phi \colon \Omega' \to \Omega
\)
by \(\Phi (\omega^{(1)}, \omega^{(2)}, \omega^{(3)}) := \omega^{(1)}\). Clearly, we have \(Q \circ \Phi^{-1} = P\) and \(Y^{(1)} = X \circ \Phi\). 
In this step, we prove that \((\llambda \otimes Q)\)-a.e. \((dY^{(2)} /d \llambda, dY^{(3)}/ d \llambda) \in \Theta  \circ \Phi\). 
By \cite[Lemma~3.2]{CN22b}, the correspondence \((t, \omega) \mapsto \Theta (t, \omega)\) is continuous with compact values, as \(F\) is compact and \(b\) and \(a\) are continuous by Condition \ref{cond: compact, LG, cont}. Additionally, compactness of \(F\) and continuity of \(b\) and \(a\) provide compactness of 
\( \Theta([t, t+1],\omega) \) for every \( (t,\omega) \in \of 0, T \gs\).
Further, Condition \ref{cond: convex} guarantees that \(\Theta\) has convex values.
Hence, \cite[Lemma 3.4]{CN22b} implies, together with \cite[Theorem 5.35]{hitchi}, that 
\begin{align} \label{eq: conseq upper hemi}
\bigcap_{m \in \mathbb{N}} \oconv \Theta ([t, t + 1/m], \omega) \subset \Theta (t, \omega)
\end{align} 
for all \((t, \omega) \in \of 0, T\gs\). Here, \(\oconv\) denotes the closure of the convex hull.
By virtue of \cite[Corollary~8, p. 48]{diestel}, \(P^n\)-a.s. for all \(t \in [0, T - 1/m]\), we have 
\begin{equation}\label{eq: P as inclusion theta}
\begin{split}
m (B^n_{t + 1/m} - B^n_t, C^n_{t + 1/m} - C^n_t) &\in \oconv ( dB^n / d \llambda, d C^n / d \llambda) ([ t, t + 1/m ]) \\&\subset \oconv \Theta ([t, t + 1/m], X).
\end{split}
\end{equation}
Thanks to Skorokhod's coupling theorem, with little abuse of notation, there exist random variables \[(X^0, B^0, C^0), (X^1, B^1, C^1), (X^2, B^2, C^2), \dots\] defined on some probability space \((\Sigma, \mathcal{G}, R)\) such that \((X^0, B^0, C^0)\) has distribution \(Q\), \((X^n, B^n, C^n)\) has distribution \(P^n\circ (X, B^n, C^n)^{-1}\) and \(R\)-a.s. \((X^n, B^n, C^n) \to (X^0, B^0, C^0)\) in the uniform topology. 
We deduce from \cite[Lemmata 3.2, 3.3]{CN22b} that the correspondence \(\omega \mapsto \Theta ([t, t + 1 /m], \omega)\) is continuous. Furthermore, as \(\oconv \Theta ([t, t + 1/m], \omega)\) is compact (by \cite[Theorem 5.35]{hitchi}) for every \(\omega \in \Omega\), it follows from \cite[Theorem 17.35]{hitchi} that the correspondence \(\omega \mapsto \oconv \Theta ([t, t + 1/m], \omega)\) is upper hemicontinuous and compact-valued. Thus, by virtue of \eqref{eq: P as inclusion theta} and \cite[Theorem 17.20]{hitchi}, we get, \(R\)-a.s. for all \(t \in [0, T - 1/m]\), that 
\[
m (B^0_{t + 1/m} - B^0_t, C^0_{t + 1/m} - C^0_t) \in \oconv \Theta ([t, t + 1/m], X^0).
\]
Notice that \((\llambda \otimes R)\)-a.e. on \(\of 0, T\of\)
\[
(d B^0 / d \llambda, d C^0 / d \llambda) = \lim_{m \to \infty} m (B^0_{\cdot + 1/m} - B^0_\cdot, C^0_{\cdot + 1/m} - C^0_\cdot).
\]
Now, with \eqref{eq: conseq upper hemi}, we get that \(R\)-a.s. for \(\llambda\)-a.a. \(t < T\)
\[
(d B^0 / d \llambda, d C^0 / d \llambda) (t) \in \bigcap_{m \in \mathbb{N}} \oconv \Theta ([t, t + 1/m], X^0) \subset \Theta (t, X^0).
\]
This shows that \( (\llambda \otimes Q)\)-a.e. \((dY^{(2)} /d \llambda, dY^{(3)}/ d \llambda) \in \Theta \circ \Phi\). 

\emph{Step 4.} In the final step of the proof, we show that \(P \in \fPas\) and we relate \((Y^{(2)}, Y^{(3)})\) to the \(P\)-semimartingale characteristics of the coordinate process.
Thanks to \cite[Lemma 11.1.2]{SV}, there exists a dense set \(D \subset \bR_+\) such that \(\tau_M \circ \Phi\) is \(Q\)-a.s. continuous for all \(M \in D\). Take some \(M \in D\). Since \(P^n \in \fPas\), it follows from the definition of the first characteristic that the process \(X_{\cdot \wedge \tau_M} - B^n_{\cdot \wedge \tau_M}\) is a local \(P^n\)-\(\F_+\)-martingale. Furthermore, by the definition of the stopping time \(\tau_M\) and the linear growth assumption, we see that \(X_{\cdot \wedge \tau_M} - B^n_{\cdot \wedge \tau_M}\) is \(P^n\)-a.s. bounded by a constant independent of~\(n\), which, in particular, implies that it is a true \(P^n\)-\(\F_+\)-martingale. Now, it follows from \cite[Proposition~IX.1.4]{JS} that \(Y^{(1)}_{\cdot \wedge \tau_M \circ \Phi} - Y^{(2)}_{\cdot \wedge \tau_M \circ \Phi}\) is a \(Q\)-\(\F'\)-martingale. Recalling that \(Y^{(2)}\) is \(Q\)-a.s. absolutely continuous by Step 2, this means that \(Y^{(1)}\) is a \(Q\)-\(\F'\)-semimartingale with first characteristic \(Y^{(2)}\). Similarly, we see that the second characteristic is given by \(Y^{(3)}\). Finally, we need to relate these observations to the probability measure \(P\) and the filtration \(\F_+\). We denote by \(A^{p, \Phi^{-1}(\F_+)}\) the dual predictable projection of some process \(A\), defined on \((\Omega', \cF')\), to the filtration \(\Phi^{-1}(\F_+)\). Recall from \cite[Lemma 10.42]{jacod79} that, for every \(t \in [0, T]\), a random variable \(Z\) on \((\Omega', \cF')\) is \(\Phi^{-1}(\cF_{t+})\)-measurable if and only if it is \(\cF^*_t\)-measurable and \(Z (\omega^{(1)}, \omega^{(2)}, \omega^{(3)})\) does not depend on \((\omega^{(2)}, \omega^{(3)})\).
Thanks to Stricker's theorem (see, e.g., \cite[Lemma~2.7]{jacod80}), \(Y^{(1)}\) is a \(Q\)-\(\Phi^{-1} (\F_+)\)-semimartingale. 
Notice that each \(\tau_M \circ \Phi\) is a \(\Phi^{-1}(\F_+)\)-stopping time and recall from Step 3 that \((\llambda\otimes Q)\)-a.e. \((d Y^{(2)}/ d \llambda, d Y^{(3)}/ d \llambda) \in \Theta\). Hence, by definition of \(\tau_M\) and the linear growth assumption, for every \(M \in D\) and \(i, j = 1, \dots, d\), we have
\[
E^Q \big[ \on{Var} (Y^{(2, i)})_{\tau_M \circ \Phi} \big] + E^Q \big[ \on{Var}(Y^{(3, ij)})_{\tau_M \circ \Phi} \big] = E^Q \Big[ \int_0^{\tau_M \circ \Phi} \Big(\Big| \frac{d Y^{(2,i)}}{d \llambda} \Big| + \Big| \frac{d Y^{(3,ij)}}{d \llambda} \Big| \Big) d \llambda \Big] < \infty,
\]
where \(\on{Var} (\cdot)\) denotes the variation process.
By virtue of this, we get from \cite[Proposition 9.24]{jacod79} that the \(Q\)-\(\Phi^{-1}(\F_+)\)-characteristics of \(Y^{(1)}\) are given by \(((Y^{(2)})^{p, \Phi^{-1}(\F_+)}, (Y^{(3)})^{p, \Phi^{-1}(\F_+)})\). 
Hence, thanks to \cite[Lemma~2.9]{jacod80}, the coordinate process \(X\) is a \(P\)-\(\F_+\)-semimartingale whose characteristics \((B^P, C^P)\) satisfy \(Q\)-a.s.
\[(B^P, C^P) \circ \Phi = ((Y^{(2)})^{p, \Phi^{-1}(\F_+)}, (Y^{(3)})^{p, \Phi^{-1}(\F_+)}).\] Consequently, we deduce from the Steps~2 and 3, and \cite[Theorem 5.25]{HWY}, that \(P\)-a.s. \((B^P, C^P) \ll \llambda\) and 
\begin{align*}
(\llambda \otimes P) \big( (d B^P / d \llambda&, d C^P / d \llambda) \not \in \Theta \big) 
\\&= (\llambda \otimes Q \circ \Phi^{-1}) \big( (d B^P / d \llambda, d C^P / d \llambda) \not \in \Theta \big)
\\&= (\llambda \otimes Q) \big( E^Q [(d Y^{(2)} / d \llambda, d Y^{(3)} / d \llambda) | \Phi^{-1} (\F_+)_-] \not \in \Theta \circ \Phi \big) = 0,
\end{align*}
where we use \cite[Corollary 8, p. 48]{diestel} for the final equality.
This means that \(P \in \cP \) and therefore, \(\cP\) is closed.

To finish the proof, it remains to show that \(\cP\) is relatively compact. Thanks to Prohorov's theorem, it suffices to prove tightness, which follows from an application of Aldous' tightness criterion as in Step~1 above. We omit the details. \qed

\subsection{Equality of \(\cM\) and \(\cD\): proof of Theorem \ref{thm: M = D}} \label{sec: pf M = D}

We prepare the proof of Theorem \ref{thm: M = D} with two auxiliary lemmata.
\begin{lemma} \label{lem: MPR meas selection}
Assume that the Conditions \ref{cond: compact, LG, cont} and \ref{SA: MPR} hold.
Take \(P \in \cP\) and denote the differential characteristics of \(X\) under \(P\) by \((b^P, a^P)\). Then, there exists a predictable function \(\f \colon \of 0, T\gs \to F\) such that \((\llambda \otimes P)\)-a.e. \((b^P, a^P)  = (a (\f) \theta (\f), a (\f)),\) where \(\theta\) is the robust MPR from Condition \ref{SA: MPR}.
\end{lemma}
\begin{proof}
    Let \(\mathscr{P}\) be the predictable \(\sigma\)-field on \(\of 0, T\gs\). Thanks to \cite[Lemma 2.9]{CN22}, the graph \(\on{gr} \Theta\) is 
\( \mathscr{P} \otimes \mathcal{B}(\bR^d) \otimes \mathcal{B}(\mathbb{S}^d_+)\)-measurable. Thus, 
\begin{align*}
G :\hspace{-0.1cm}&= \big\{ (t, \omega) \in \of 0, T\gs \colon (b^P_t (\omega), a^P_t (\omega)) \not \in \Theta ( t, \omega )\big\}
\\&= \big\{ (t, \omega) \in \of 0, T \gs \colon (t,\omega, b^P_t (\omega), a^P_t (\omega)) \not \in \on{gr} \Theta \big\} \in \mathscr{P}.
\end{align*}
We define 
\[
\pi (t, \omega) := \begin{cases} (b (f_0, t, \omega), a (f_0, t, \omega)),& \text{if } (t, \omega) \in G,\\
(b^P_t (\omega), a^P_t (\omega)),& \text{if } (t, \omega) \not \in G, \end{cases}
\]
where \(f_0 \in F\) is arbitrary but fixed. 
Thanks to the measurable implicit function theorem \cite[Theorem~18.17]{hitchi}, as \((b, a)\) is a Carath\'eodory function on \(F \times \of 0, T\gs\) in the sense that it is continuous in the \(F\) and \(\mathscr{P}\)-measurable in the \(\of 0, T\gs\) variable, the correspondence \(\gamma \colon \of 0, T\gs \twoheadrightarrow F\) defined by 
\[
\gamma (t, \omega) := \big\{ f \in F \colon (b (f, t, \omega), a (f, t, \omega)) = \pi (t, \omega) \big\}
\]
is \(\mathscr{P}\)-measurable and it admits a measurable selector, i.e.,  there exists a \(\mathscr{P}\)-measurable function \(\f \colon \of 0, T \gs \to F\) such that \(\pi (t, \omega) = (b (\f(t, \omega), t, \omega), a (\f(t, \omega), t, \omega))\) for all \((t, \omega) \in \of 0, T\gs\). Since \(P \in \cP\), we have \((\llambda \otimes P)\)-a.e. \(\pi = (b^P, a^P)\), and further \(b = a \theta\) by Condition \ref{SA: MPR}. Putting these pieces together, we conclude that \(\f\) has all claimed properties.
\end{proof}

The second lemma can be seen as an extension of Bene\u{s}' condition (\cite[Corollary~3.5.16]{KaraShre}). To prove the lemma we use a local change of measure in combination with a Gronwall type argument (see, e.g., \cite{CFY,C20} for related strategies). 
\begin{lemma} \label{lem: girsanov}
Let \(P \in \fPas\), denote the differential characteristics of \(X\) under \(P\) by \((b^P, a^P)\) and define the continuous local \(P\)-martingale part of the coordinate process \(X\) by
\[
X^c := X - X_0 - \int_0^\cdot b^P_s ds.
\]
Further, let \(c^P\) be a predictable process.
Assume the following three conditions:
\begin{enumerate}
    \item[\textup{(a)}]
    For every \(N \in \mathbb{N}\) there exists a constant \(C = C_N > 0\) such that \(P\)-a.s.
\begin{align} \label{eq: loc bdd QV}
\int_0^{T_N} \langle c^P_s, a^P_s c^P_s \rangle ds \leq C,
\end{align}
where
\[
T_N = \inf \{ t \in [0, T] \colon \|X_t\| \geq N\} \wedge T.
\]
\item[\textup{(b)}]
  There exists a constant \(C > 0\) such that \(P\)-a.s. for \(\llambda\)-a.a. \(t \in [0, T]\)
\begin{align*}
\|b^P_t + a^P_t c^P_t\|^2 + \on{tr} \big[ a^P_t \big] \leq C \Big(1 +  \sup_{s \in [0, t]} \|X_s\|^2 \Big).
\end{align*}  
\item[\textup{(c)}] There exists a constant \(C > 0\) such that \(P\)-a.s. \(\|X_0\| \leq C\).
\end{enumerate}
Then, the stochastic exponential
\[
Z^{P} := \exp \Big( \int_0^\cdot \langle c^P_s , d X^c_s \rangle - \frac{1}{2} \int_0^\cdot \langle c^P_s, a^P_s c^P_s \rangle ds \Big)
\]
is a well-defined \(P\)-martingale.
\end{lemma}
\begin{proof}
For a moment, we fix \(N \in \mathbb{N}\). Thanks to the assumption (a), Novikov's condition implies that the stopped process \(Z^P_{\cdot \wedge T_N}\) is a \(P\)-martingale and the global process \(Z^P\) is a well-defined, non-negative local \(P\)-martingale, i.e., in particular a \(P\)-supermartingale. Thus, \(Z^P\) is a \(P\)-martingale if and only if \(E^P [Z^P_T] = 1\). In the following we prove this property.
Define a probability measure \(Q_N\) via the Radon--Nikodym density \(d Q_N / d P = Z^P_{T \wedge T_N}\). As \(Q_N \sim P\),
Girsanov's theorem (\cite[Theorem~III.3.24]{JS}) yields that \(X\) is a \(Q_N\)-semimartingale with absolutely continuous characteristics whose densities \((b^{Q_N}, a^{Q_N})\) are given by 
\[
b^{Q_N} = b^P + a^P c^P \1_{\of 0, T_N\gs}, \quad a^{Q_N} = a^P.
\]
By assumption (b) and the equivalence \(Q_N \sim P\), there exists a constant \(C > 0\) such that, \(Q_N\)-a.s. for \(\llambda\)-a.a. \(t \in [0, T_N]\), we have
\begin{align} \label{eq: linear growth modi drift}
\|b^{Q_N}\|^2 + \on{tr} \big[ a^{Q_N}_t \big] \leq C \Big(1 +  \sup_{s \in [0, t]} \|X_s\|^2 \Big).
\end{align}
Now, using standard arguments (see \cite[pp. 389--390]{KaraShre}), hypothesis (c) and \eqref{eq: linear growth modi drift}, we get, for all \(t \in [0, T]\), that
\begin{align*}
E^{Q_N} \Big[ \sup_{s \in [0, t \wedge T_N]} \|X_s\|^2 \Big] 
&\leq C \Big( 1 + E^{Q_N} \Big[ \int_0^{t \wedge T_N} \big( \|b^{Q_N}_s\|^2 + \on{tr} \big[ a^{Q_N}_s\big] \big) ds \Big] \Big)
\\&\leq C \Big( 1 + \int_0^t E^{Q_N} \Big[ \sup_{r \in [0, s \wedge T_N]} \|X_r\|^2 \Big] ds \Big).
\end{align*}
where the constant \(C > 0\) is independent of \(N\). Gronwall's lemma yields that 
\[
E^{Q_N} \Big[ \sup_{s \in [0, t \wedge T_N]} \|X_s\|^2 \Big] \leq C e^{C T}, \qquad t \in [0, T].
\]
Hence, by Chebyshev's inequality, we get that 
\begin{align*}
Q_N (T_N \leq T) = Q_N \Big( \sup_{s \in [0, T \wedge T_N]} \|X_s\| \geq N \Big)
 \leq  \frac{ C e^{C T} }{N^2} \to 0 \text{ with } N \to \infty.
\end{align*}
Finally, using the monotone convergence theorem for the first equality, we obtain
\begin{align*}
E^P \big[ Z^P_T \big] &= \lim_{N \to \infty} E^P \big[ Z^P_t \1_{\{T_N > T\}} \big] = \lim_{n \to \infty} Q_N (T_N > T) = 1,
\end{align*}
which completes the proof.
\end{proof}

\begin{proof}[Proof of Theorem \ref{thm: M = D}]
Take \(P \in \cP\) and denote the differential characteristics of \(X\) under \(P\) by \((b^P, a^P)\). By Lemma \ref{lem: MPR meas selection}, there exists a predictable function \(\f\) such that \((\llambda \otimes P)\)-a.e. \((b^P, a^P) = (b (\f), a(\f))\). Let \(\theta\) be the robust MPR from Condition \ref{SA: MPR} and 
define 
\begin{align} \label{eq: ZP}
Z^{P} := \exp \Big(  - \int_0^\cdot \langle \theta (\f_s) , d X^c_s \rangle - \frac{1}{2} \int_0^\cdot \langle \theta (\f_s), a(\f_s) \theta (\f_s) \rangle ds \Big).
\end{align}
Using the Conditions \ref{cond: compact, LG, cont} and \ref{SA: MPR}, it follows from Lemma \ref{lem: girsanov} that \(Z^P\) is a \(P\)-martingale. Hence, we may define a probability measure \(Q \sim P\) via the Radon--Nikodym derivative \(dQ/dP = Z^P_T\).
By Girsanov's theorem (\cite[Theorem~III.3.24]{JS}), \(Q \in \fPas\) and the differential characteristics of \(X\) under \(Q\) are given by \((b (\f) - a (\f) \theta (\f), a(\f)) = (0, a (\f)) \in \tilde{\Theta}\). In particular, this shows that \(Q \in \M\).

Conversely, take \(Q \in \cM\) and let \(a^Q\) be the second differential characteristic of \(X\) under~\(Q\). 
By Lemma~\ref{lem: MPR meas selection}, which we can use because the zero function is a feasible choice for the coefficient~\(b\), there exists a predictable function
\( \f \colon \of 0, T\gs \hspace{0.1cm}\to F \) 
such that \((\llambda \otimes Q)\)-a.e. \(a^Q = a (\f)\). Let \(\theta\) be the robust MPR from Condition \ref{SA: MPR} and define
\[
Z^{Q} := \exp \Big(  \int_0^\cdot \langle \theta (\f_s) , d X_s \rangle - \frac{1}{2} \int_0^\cdot \langle \theta(\f_s), a(\f_s) \theta(\f_s) \rangle ds \Big).
\]
Using the Conditions \ref{cond: compact, LG, cont} and \ref{SA: MPR}, we deduce from Lemma \ref{lem: girsanov} that \(Z^Q\)
is a \(Q\)-martingale. Therefore, we can define a measure \(P \sim Q\) via the Radon--Nikodym derivative \(dP / dQ = Z^Q_T\). As in the previous case, we deduce from Girsanov's theorem that \(P \in \fPas\) with differential characteristics
\((a(\f) \theta(\f), a(\f)) 
= (b (\f), a (\f)) \in \Theta\). We conclude that \(P \in \cP\).
\end{proof}

\subsection{Duality theory for robust utility maximization: proofs of Theorems \ref{thm: main positive power exponential} and \ref{thm: main negative power log}}
\label{sec: proof duality utility}
The idea of proof is to apply the abstract duality results given by \cite[Theorems 2.10 and 2.16]{kupper}. This requires some care to account for the lack of boundedness from above of the power utility \(U(x) = \frac{x^p}{p}\), \( p \in (0,1)\), and the log utility \(U(x) = \log(x)\).

\subsubsection{Some preparations}
The following lemma is a generalization of \cite[Corollary 2]{GriMa03} in the sense that, instead of Brownian motion, we consider a continuous local martingale with uniformly elliptic volatility. The novelty in our proof is the application of time change and comparison arguments to deduce certain moment bounds for the driving local martingale from those of Brownian motion.

\begin{proposition} \label{prop: integrability of stochastic exponential}
Suppose that Condition \ref{cond: unif ellipticity vola} holds.
Let \(Q \in \cM\) and denote the differential characteristics of \(X\) under \(P\) by \((b^Q = 0, a^Q)\). Furthermore, take a predictable process \(c^Q\) of linear growth, i.e., such that there exists a constant \(\C > 0\) such that
\[
\|c_t^Q (\omega)\| \leq \C \Big(1 + \sup_{s \in [0, t]} \|\omega (s)\|\Big)
\]
for \((\llambda \otimes Q)\)-a.a. \((t, \omega) \in \of 0, T\gs\). Define a continuous local \(P\)-martingale by
\[
Z^Q := \exp \Big( \int_0^\cdot \langle c^Q_s, d X_s\rangle - \frac{1}{2} \int_0^\cdot \langle c^Q_s, a^Q_s c^Q_s \rangle ds \Big).
\]
For every \(p \geq 1\), we have 
\[
E^Q \big[ (Z^Q_T)^p \big] < \infty.
\]
\end{proposition}
\begin{proof}
Throughout the proof, fix \(p \geq 1\).
By virtue of \cite[Corollary 1]{GriMa03}, it suffices to prove that there exists a partition \(0 = t_0 < t_1 < \ldots < t_m = T\) of the interval \([0, T]\) such that 
\[
E^Q \Big[ \exp \Big( \C_p \int_{t_{n - 1}}^{t_n} \langle c^Q_s, a^Q_s c^Q_s\rangle ds \Big) \Big] < \infty, \quad n = 1, 2, \dots, m.
\]
Fix \(n \in \{1, \dots, m\}\).
By the linear growth assumption on \(c^Q\) and the \((\llambda \otimes Q)\)-a.e. boundedness assumption on \(a^Q\) (which stems from the definition of \(\cM\) and Condition \ref{cond: unif ellipticity vola}), we have
\begin{align*}
    E^Q \Big[ \exp \Big( \C_p \int_{t_{n - 1}}^{t_n} \langle c^Q_s, a^Q_s c^Q_s\rangle ds \Big) \Big] &\leq \C E^Q \Big[ \exp \Big( \C_p (t_n - t_{n - 1}) \sup_{s \in [0, T]} \|X_s\|^2 \Big) \Big],
\end{align*}
where \(\C_p > 0\) depends on \(T > 0\) and the power \(p\).
It\^o's formula shows that 
\[
d \|X_t\|^2 = 2 \langle X_t, d X_t \rangle + \on{tr} \big[ a^Q_t \big] dt.
\]
Thanks to Condition \ref{cond: unif ellipticity vola}, there exists a constant \(\K \in \mathbb{N}\) such that 
\begin{align}\label{eq: ellipticity constant}
\frac{\|\xi\|^2}{\K} \leq \langle \xi, a (f, t, \omega) \xi \rangle \leq \K \|\xi\|^2 
\end{align}
for all \((\xi, f, t, \omega) \in \bR^d \times F \times \of 0, T \gs\). 
Next, define 
\[
L := \int_0^\cdot \Big[\frac{\langle X_s, a^Q_s X_s \rangle }{ \K \|X_s\|^2} \1_{\{X_s \not = 0\}} + \1_{\{X_s = 0\}} \Big] ds,
\]
and \(S_t := \inf \{s \in [0, T] \colon L_s \geq t \}\) for \(t \in [0, L_T]\). Notice from \eqref{eq: ellipticity constant} that \(L\) is strictly increasing, continuous and \(L_T \leq T\). 
Hence, \(S\) is continuous and the inverse of \(L\).
In the following we use standard results from \cite[Section~V.1]{RY} on time changed continuous semimartingales without explicitly mentioning them. We obtain that, for \(t \in [0, L_T]\),
\[
d \|X_{S_t}\|^2 = 2 \langle X_{S_t}, d X_{S_t} \rangle + \on{tr} \big[ a^Q_{S_t} \big] d S_t.
\]
Further, we obtain that, for \(t \in [0, L_T]\), 
\begin{align*}
\int_0^{t} \Big[\frac{\K \|X_{S_s}\|^2}{\langle X_{S_s}, a^Q_{S_s} X_{S_s}\rangle} \1_{\{X_{S_s} \not = 0\}} + \1_{\{X_{S_s} = 0\}} \Big] ds &= \int_0^{t} \Big[\frac{\K \|X_{S_s}\|^2}{\langle X_{S_s}, a^Q_{S_s} X_{S_s}\rangle} \1_{\{X_{S_s} \not = 0\}} + \1_{\{X_{S_s} = 0\}} \Big] dL_{S_s}
\\&= \int_0^{S_t} \Big[ \frac{\K \|X_s\|^2}{\langle X_s, a^Q_s X_s\rangle} \1_{\{X_s \not = 0\}} + \1_{\{X_s = 0\}}\Big] dL_s = S_t.
\end{align*}
Hence, 
\[
\1_{[0, L_T]} (s) d S_s = \1_{[0, L_T]} (s) \Big[\frac{\K \|X_{S_s}\|^2}{\langle X_{S_s}, a^Q_{S_s} X_{S_s}\rangle} \1_{\{X_{S_s} \not = 0\}} + \1_{\{X_{S_s} = 0\}} \Big] ds,
\]
which implies, for \(t \in [0, L_T]\), that 
\[
d \|X_{S_t}\|^2 = 2 \langle X_{S_t}, d X_{S_t} \rangle + \on{tr} \big[ a^Q_{S_t} \big] \Big[\frac{\K \|X_{S_t}\|^2}{\langle X_{S_t}, a^Q_{S_t} X_{S_t}\rangle} \1_{\{X_{S_t} \not = 0\}} + \1_{\{X_{S_t} = 0\}} \Big] dt.
\]
Notice that \(\int_0^{\cdot \wedge L_T} \langle X_{S_t}, d X_{S_t}\rangle\) is a continuous local martingale (for a time-changed filtration) with second characteristic
\begin{align*}
\int_0^{\cdot \wedge L_T} \langle X_{S_s}, a^Q_{S_s} X_{S_s} \rangle dS_s 
&= \int_0^{\cdot \wedge L_T}  \frac{ \langle X_{S_s}, a^Q_{S_s} X_{S_s} \rangle \K \|X_{S_s}\|^2}{\langle X_{S_s}, a^Q_{S_s} X_{S_s}\rangle} \1_{\{X_{S_s} \not = 0\}} ds
= \int_0^{\cdot \wedge L_T} \K \|X_{S_s}\|^2 ds.
\end{align*}
By a classical representation theorem for continuous local martingales (see \cite[Theorem III.7.1\('\), p.~90]{IW}), on a standard extension of the underlying filtered probability space, there exists a one-dimensional standard Brownian motion \(W\) such that, for all \(t \in [0, L_T]\),
\[
d \|X_{S_t}\|^2 = 2 \sqrt{\K} \|X_{S_t}\| d W_t + \on{tr} \big[ a^Q_{S_t} \big] \Big[\frac{\K \|X_{S_t}\|^2}{\langle X_{S_t}, a^Q_{S_t} X_{S_t}\rangle} \1_{\{X_{S_t} \not = 0\}} + \1_{\{X_{S_t} = 0\}} \Big] dt.
\]
By virtue of \eqref{eq: ellipticity constant}, we get that \(Q\)-a.s. for \(\llambda\)-a.a. \(t \in [0, L_T]\)
\begin{align} \label{eq: bound drift}
\on{tr} \big[ a^Q_{S_t} \big] \Big[ \frac{\K \|X_{S_t}\|^2}{\langle X_{S_t}, a^Q_{S_t} X_{S_t}\rangle} \1_{\{X_{S_t} \not = 0\}} + \1_{\{X_{S_t} = 0\}} \Big] 
& \leq \on{tr} \big[ a^Q_{S_t} \big] \big[ \K^2 \1_{\{X_{S_t} \not = 0\}} + \1_{\{X_{S_t} = 0\}} \big]
\leq d \K^3.
\end{align}
Let \(Y\) be a continuous semimartingale with dynamics
\begin{align} \label{eq: SDE Y}
d Y_t = 2 \sqrt{\K |Y_t|} d W_t + d \K^{3}\hspace{0.05cm} dt, \quad Y_0 = \|x_0\|^2.
\end{align}
Such a process exists as its SDE satisfies strong existence (see, e.g., \cite[Chapter IX]{RY} or \cite[Chapter~5]{KaraShre}). Furthermore, as the SDE 
\[
d Z_t = 2 \sqrt{\K |Z_t|} d W_t, \quad Z_0 = 0,
\]
has the (up to indistinguishability) unique solution \(Z \equiv 0\), it follows from \cite[Proposition IX.3.6]{RY} that \(Q\)-a.s. \(Y \geq 0\).
Next, we use a comparison argument as in the proofs of \cite[Theorem IX.3.7]{RY} or \cite[Lemma 5.6]{C20} to relate the processes \(\|X_S\|^2\) and \(Y\).
Notice that \(Q\)-a.s. for all \(t \in [0, T]\)
\begin{align*}
\int_0^{t \wedge L_T} &\frac{\1_{\{Y_s < \|X_{S_s}\|^2\}}}{4 \K |\|X_{S_s}\|^2 - Y_s|} d [ \|X_{S}\|^2 - Y,\|X_S\|^2 - Y]_s
\\&\qquad = \int_0^{t \wedge L_T}\frac{\1_{\{Y_s < \|X_{S_s}\|^2\}}}{4 \K |\|X_{S_s}\|^2 - Y_s|} 4 \K \big(\|X_{S_s}\| - \sqrt{Y_s}\big)^2 ds 
\\&\qquad \leq \int_0^{t \wedge L_T}\frac{\1_{\{Y_s < \|X_{S_s}\|^2\}}}{|\|X_{S_s}\|^2 - Y_s|} |\|X_{S_s}\|^2 - Y_s| ds \leq t.
\end{align*}
Hence, by \cite[Lemma IX.3.3]{RY}, \(Q\)-a.s. \(L^0_{\cdot \wedge L_T} (\|X_S\|^2 - Y) = 0\), where \(L^0\) denotes the semimartingale local time in zero. 
Using this observation, Tanaka's formula and \eqref{eq: bound drift} yield that \(Q\)-a.s. for all \(t \in [0, L_T]\)
\begin{align*}
\big(\|X_{S_t}\|^2 - Y_t\big)^+ &= \int_0^t \1_{\{Y_s < \|X_{S_s}\|^2\}} d \big( \|X_{S_s}\|^2 - Y_s \big)
\\&\leq \int_0^t \1_{\{Y_s < \|X_{S_s}\|^2\}} 2 \K \big[\|X_{S_s}\| - \sqrt{Y_s} \big] d W_s.
\end{align*}
As the coefficients of the SDEs for \(Y\) and \(\|X_{S_{\cdot \wedge L_T}}\|\) satisfy standard linear growth conditions, these processes have polynomial moments and it follows readily that the It\^o integral process
\[
\int_0^{\cdot \wedge L_T} \1_{\{Y_s < \|X_{S_s}\|^2\}} 2 \K \big[\|X_{S_s}\| - \sqrt{Y_s} \big] d W_s
\]
is a martingale. Consequently, for all \(t \in [0, T]\),
\[
E^Q \Big[ \big(\|X_{S_{t \wedge L_T}}\|^2 - Y_{t \wedge L_T}\big)^+ \Big] = 0.
\]
By the continuous paths of \(Y\) and \(X_{S_{\cdot \wedge L_T}}\), we conclude that \(Q\)-a.s. \(Y_t \geq \|X_{S_t}\|^2\) for all \(t \in [0, L_T]\).
Let \(B = (B^{(1)}, \dots, B^{(d \K^2)})\) be a \(d \K^2\)-dimensional standard Brownian motion such that \(\|B_0\|^2 = \|x_0\|^2\).
By L\'evy's characterization of Brownian motion, the process
\[
\overline{B} := \sum_{k = 1}^{d \K^2} \int_0^\cdot \frac{B^{(k)}_{\K s}}{\sqrt{\K}\|B_{\K s}\|} d B^{(k)}_{\K s}
\]
is a one-dimensional standard Brownian motion, and, by It\^o's fomula, 
\[
d \|B_{\K t}\|^2 = 2\sqrt{\K} \|B_{\K t}\| d \overline{B}_t + d \K^3 dt.
\]
As the SDE \eqref{eq: SDE Y} satisfies uniqueness in law (see, e.g., \cite[Chapter IX]{RY} or \cite[Chapter~5]{KaraShre}), we conclude that \(Y = \|B_{\K \cdot}\|^2\) in law. 
For the remainder of this proof, we presume that the partition \(t_1, \dots, t_m\) is choosen such that \(t_n - t_{n - 1} < 1/(2\C_p \K T)\) for all \(n = 1, \dots, m\). Then, by \cite[Proposition 1.3.6]{KaraShre}, the process \((\exp ( \C_p (t_n - t_{n - 1}) \|B_{\K s}\|^2))_{s \in [0, T]}\) is a positive submartingale. 
Using that \(L_T \leq T\) and Doob's maximal inequality, we obtain
\begin{align*}
    E^Q \Big[ \exp \Big( \C_p (t_n - t_{n - 1}) \sup_{s \in [0, T]} \|X_s\|^2 \Big) \Big] &= E^Q \Big[ \exp \Big( \C_p (t_n - t_{n - 1}) \sup_{s \in [0, T]} \|X_{S_{L_s}}\|^2 \Big) \Big]
    \\&\leq E^Q \Big[ \exp \Big( \C_p (t_n - t_{n - 1}) \sup_{s \in [0, L_T]} \|Y_{s}\|^2 \Big) \Big]
    \\&\leq E^Q \Big[ \exp \Big( \C_p (t_n - t_{n - 1}) \sup_{s \in [0, T]} \|Y_{s}\|^2 \Big) \Big]
    \\&= E \Big[ \sup_{s \in [0, T]} \exp \Big( \C_p (t_n - t_{n - 1}) \|B_{\K s}\|^2 \Big) \Big]
    \\&\leq  4 E \Big[ \exp \Big( \C_p (t_n - t_{n - 1}) \|B_{\K T}\|^2 \Big) \Big] < \infty.
\end{align*}
The proof is complete.
\end{proof}

\begin{lemma} \label{lem: moments density}
Assume that the Conditions \ref{cond: compact, LG, cont} and \ref{SA: MPR} hold. Additionally,
suppose that either Condition~\ref{cond: MPR bounded} or Condition \ref{cond: unif ellipticity vola} holds.
Then, for every \(P \in \cP\) there exists a probability measure \(Q \in \fM_e(\cP)\) such that
\begin{align} \label{eq: moment bound density}
E^Q \Big[ \Big( \frac{dP}{dQ} \Big)^p \Big] < \infty, \quad \forall\hspace{0.05cm} p > 0.
\end{align}
\end{lemma}
\begin{proof}
Take \(P \in \cP\) and denote the differential characteristics of \(X\) under \(P\) by \((b^P, a^P)\). By Lemma \ref{lem: MPR meas selection}, there exists a predictable function \(\f \colon \of 0, T\gs \to F\) such that \((\llambda \otimes P)\)-a.e. \((b^P, a^P) = (b (\f), a (\f))\).
Let \(Z^P\) be as in \eqref{eq: ZP} and recall that it is a \(P\)-martingale by Lemma \ref{lem: girsanov}. We define a probability measure \(Q \sim P\) by the Radon--Nikodym derivative \(dQ/dP = Z^P_T\).
Then, Girsanov's theorem (\cite[Theorem~III.3.24]{JS}) shows that \(Q \in \fM_e(\cP)\)
and simple computation yield that \(Q\)-a.s.
\[
\frac{d P}{d Q} = \Big( \frac{d Q}{d P} \Big)^{-1} = \exp \Big( \int_0^T \langle \theta (\f_s) , d X_s \rangle - \frac{1}{2} \int_0^T \langle \theta (\f_s), a(\f_s) \theta (\f_s) \rangle ds \Big).
\]
In case Condition \ref{cond: MPR bounded} holds, \(\int_0^T \langle \theta (\f_s), a(\f_s) \theta (\f_s) \rangle ds\) is \(Q\)-a.s. bounded and \eqref{eq: moment bound density} follows from \cite[Theorem~1]{GriMa03}.
Further, if Condition \ref{cond: unif ellipticity vola} holds, \(\theta (\f) = a^{-1} (\f) b (\f)\) is of linear growth by Condition~\ref{cond: compact, LG, cont}, and Proposition~\ref{prop: integrability of stochastic exponential} yields \eqref{eq: moment bound density}. This completes the proof.
\end{proof}

\begin{lemma} \label{lem: bound sequence}
Assume that the Conditions \ref{cond: compact, LG, cont} and \ref{SA: MPR} hold. Additionally,
suppose that either Condition~\ref{cond: MPR bounded} or Condition \ref{cond: unif ellipticity vola} holds.
Let \( x > 0\) and \( (g_n)_{n \in \bN} \subset \cC(x)\). Then, for every \(P \in \cP\)
and every \( \epsilon \in (0,1)\), we have
\[
\sup_{n \in \bN} E^P \big[ (g_n)^\epsilon \big] < \infty.
\]
\end{lemma}
\begin{proof}
The lemma follows similar to \cite[Lemma 5.11]{kupper}.
Fix \(\epsilon \in (0,1)\), \( P \in \cP\) and let \( (g_n)_{n \in \bN} \subset \cC(x)\). Set \( p := \frac{1}{1-\epsilon} \). By virtue of Lemma \ref{lem: moments density}, there exists a probability measure \(Q \in \fM_e(\cP)\) with 
\[
E^Q \Big[ \Big( \frac{dP}{dQ}  \Big)^p \Big] < \infty.
\]
Hence, Hölder's inequality together with Proposition \ref{prop: first duality} implies
\[ \sup_{n \in \bN} E^P \big[ (g_n)^\epsilon \big] \leq E^Q \Big[ \Big( \frac{dP}{dQ} \Big)^p \Big]^{1-\epsilon} E^Q \big[ g_n\big]^\epsilon 
\leq E^Q \Big[ \Big( \frac{dP}{dQ} \Big)^p \Big]^{1-\epsilon} x^\epsilon < \infty,
\]
which gives the claim.
\end{proof}

The following estimate can be extracted from the proof of \cite[Theorem 2.10]{kupper}.

\begin{lemma} \label{lem: weak duality}
Let \(P \in \cP\), and let \(Q \in \cD\) be such that \(Q \ll P\).
Then, for every \(x, y > 0\), we have
\begin{equation} \label{eq: weak duality}
    u(x) \leq \overline{u}(x) \leq E^P \Big[ \max\Big\{ V_1\Big(y \frac{dQ}{dP} \Big), 0 \Big\}\Big] + xy.
\end{equation}
\end{lemma}

Finally, we present two lemmata which deal with the power and the log utility separately.
\begin{lemma} \label{lem: power auxiliary}
Assume that the Conditions \ref{cond: compact, LG, cont} and \ref{SA: MPR} hold. Additionally,
suppose that either Condition~\ref{cond: MPR bounded} or Condition \ref{cond: unif ellipticity vola} holds.
Let \(U\) be a power utility function \(U(x) = \frac{x^p}{p}\) with exponent \( p \in (0,1)\). Then, 
\begin{enumerate}
\item[\textup{(i)}] there exists an \( x > 0\) such that \( \overline{u}(x) < \infty\),
    
\item[\textup{(ii)}] for every \(x > 0\) and \( (g_n)_{n \in \bN} \subset \cC(x)\), the sequence of random variables 
\[ \max \big\{ U\big(g_n + \tfrac{1}{n}\big), 0 \big\}, \quad n \in \bN, \]
is uniformly integrable for every \( P \in \cP\).
\end{enumerate}
\end{lemma}
\begin{proof}
We start with (i). Let \(P \in \cP\).
By virtue of Lemma \ref{lem: weak duality}, it suffices to construct \(Q \in \cD\) such that \(Q \ll P\) and 
\[ 
E^P \Big[ \max\Big\{ V_1\Big(y \frac{dQ}{dP} \Big), 0 \Big\}\Big] < \infty,
\qquad 
V_1(y) = \sup_{x \geq 0} \Big[ \frac{(x+1)^p}{p} - xy \Big].
\]
This follows as in \cite[Lemma 5.16]{kupper}, when replacing \cite[Lemma 5.10]{kupper} by Lemma \ref{lem: moments density}.
Regarding (ii), this can be shown as \cite[Lemma 5.13]{kupper} by using Lemma \ref{lem: bound sequence} instead of \cite[Lemma 5.11]{kupper}.
\end{proof}

\begin{lemma} \label{lem: log auxiliary}
Assume that the Conditions \ref{cond: compact, LG, cont} and \ref{SA: MPR} hold. Additionally,
suppose that either Condition~\ref{cond: MPR bounded} or Condition \ref{cond: unif ellipticity vola} holds.
Let \(U\) be the log utility \(U(x) = \log(x)\). Then,
\begin{enumerate}
    \item[\textup{(i)}] there exists an \( x > 0\) such that \( \overline{u}(x) < \infty\),
    
    \item[\textup{(ii)}] for every \(x > 0\) and \( (g_n)_{n \in \bN} \subset \cC(x)\), the sequence of random variables 
    \[ \max \big\{ U \big(g_n + \tfrac{1}{n} \big), 0 \big \}, \quad n \in \bN, \]
    is uniformly integrable for every \( P \in \cP\),
    
    \item[\textup{(iii)}]  for each \( y > 0\), and each \( P \in \cP\), there exists \( Q \in \cD\) with \(Q \ll P\) such that
    \[ E^P \Big[ \max \Big\{ V_1 \Big(y \frac{dQ}{dP}\Big), 0 \Big\} \Big] < \infty, \]
    where
    \[
    V_1(y) := \sup_{x \geq 0} \big[ \log(x+1) - xy \big].
    \]
\end{enumerate}
\end{lemma}
\begin{proof}
By virtue of Lemma \ref{lem: weak duality}, (iii) implies (i).
To see (iii), one argues as in \cite[Lemma 5.15]{kupper}, replacing 
\cite[Lemma 5.10]{kupper} by Lemma \ref{lem: moments density}.
Regarding (ii), this can be shown as \cite[Lemma 5.12]{kupper}, using 
Lemma~\ref{lem: bound sequence} instead of \cite[Lemma 5.11]{kupper}.
\end{proof}

\subsubsection{Duality for utilities bounded from below: proof of Theorem \ref{thm: main positive power exponential}}
Recall that in this section \(U\) is either a \emph{power utility} \(U(x) = \frac{x^p}{p}\), \( p \in (0,1)\), or an \emph{exponential utility} \(U(x) = -e^{-\lambda x}\), \( \lambda >0\).
Note that both utilities are bounded from below, and that the power utility is unbounded from above.
Corollary~\ref{cor: M = D} implies that \NFLVR holds. Hence, we deduce from Theorem \ref{thm: dualities} that \(\cC\) and \(\cD\) are in duality and that \(\cD = \fM_a(\cP)\).
Applying Corollary \ref{cor: M = D} once more, Theorem \ref{thm: compactness} shows that the sets 
\(\cP\) and \(\cD = \cM\) are convex and compact.
Using Corollary \ref{cor: M = D} a third time proves that \eqref{eq: cond 2.1} holds.
Hence, by virtue of \cite[Theorem 2.10]{kupper}, the claim follows directly in case of a exponential utility. To handle the power utility, we additionally apply Lemma \ref{lem: power auxiliary}.\qed

\subsubsection{Duality for utilities unbounded from below: proof of Theorem \ref{thm: main negative power log}}
Recall that in this section \(U\) is either a
\emph{power utility} \(U(x) = \frac{x^p}{p}\), \( p \in (-\infty,0)\),
or the \emph{log utility} \(U(x) = \log(x)\).
Note that both utilities are unbounded from below, i.e., \(U(0) = \lim_{x \to 0} U(x) = -\infty\), and that the log utility is unbounded from above.
Corollary \ref{cor: M = D} implies that \NFLVR holds. Hence, we deduce from Theorem \ref{thm: dualities} that \(\cC\) and \(\cD\) are in duality and that \(\cD = \fM_a(\cP)\).
Applying Corollary \ref{cor: M = D} once more, Theorem \ref{thm: compactness} shows that the sets 
\(\cP\) and \(\cD = \cM\) are convex and compact.
Using Corollary \ref{cor: M = D} a third time proves that \eqref{eq: cond 2.1} holds.
Hence, by virtue of \cite[Theorem 2.16]{kupper}, the claim follows directly in case of a power utility.
To handle the log utility, we additionally apply Lemma \ref{lem: log auxiliary}.\qed





\begin{thebibliography}{1}

\bibitem{hitchi}
C.~D.~Aliprantis and K.~B.~Border. 
\newblock {\em Infinite Dimensional Analysis: A Hitchhiker's Guide}.
\newblock Springer Berlin Heidelberg, 3rd ed., 2006.

\bibitem{bartl}
D.~Bartl, P.~Cheridito and M.~Kupper.
\newblock Robust expected utility maximization with medial limits.
\newblock {\em Journal of Mathematical Analysis and Applications}, 471(3):752--775, 2019.

\bibitem{kupper}
D.~Bartl, M.~Kupper and A.~Neufeld.
\newblock Duality theory for robust utility maximisation.
\newblock {\em Finance and Stochastics}, 25(3):469--503, 2021.

\bibitem{biagini}
S.~Biagini and M.~{\c{C}}.~P{\i}nar.
\newblock The robust Merton problem of an ambiguity averse investor.
\newblock {\em Mathematics and Financial Economics}, 11(1):1--24, 2017.

\bibitem{BC18}
R.~Blanchard and L.~Carassus.
\newblock Multiple-priors optimal investment in discrete time for unbounded utility functions.
\newblock {\em The Annals of Applied Probability}, 28(3):1856--1892, 2018.

\bibitem{bogachev}
V.~I.~Bogachev.
\newblock {\em Measure Theory}.
\newblock Springer Berlin Heidelberg, 2007.

\bibitem{nutz_nondom}
B.~Bouchard and M.~Nutz.
\newblock Arbitrage and Duality in Nondominated Discrete-Time Models.
\newblock {\em Annals of Applied Probability}, 25(2):823--859, 2015.

\bibitem{BR16}
C.~Bruggeman and J.~Ruf.
\newblock A one-dimensional diffusion hits points fast.
\newblock {\em Electronic Communications in Probability}, 21(22):1--7, 2016.

\bibitem{COW19}
L.~Carassus, J.~Ob\l{}\'{o}j and J.~Wiesel.
\newblock The robust superreplication problem: A dynamic approach.
\newblock {\em SIAM Journal on Financial Mathematics}, 10(4):907--941, 2019.

\bibitem{CFY}
P.~Cheridito, D.~Filipovi\'c and M.~Yor
\newblock Equivalent and absolutely continuous measure changes for jump-diffusion processes.
\newblock {\em The Annals of Applied Probability}, 15(3):1713--1732, 2005.

\bibitem{ChuStr}
T.~Choulli and C.~Stricker.
\newblock Deux applications de la d\'ecompositionde Glatchouk--Kunita--Watanabe.
\newblock {\em S{\'e}minaire de Probabilit{\'e}s XXX}, pp. 12–23, Lecture Notes in Mathematics vol. 1626, Springer, 1996.

\bibitem{CH89}
J.~C.~Cox and C.~F.~Huang.
\newblock Optimal consumption and portfolio policies when asset prices follow a diffusion process.
\newblock {\em Journal of Economic Theory}, 49:33--83, 1989.

\bibitem{CH91}
J.~C.~Cox and C.~F.~Huang.
\newblock A variational problem arising in financial economics.
\newblock {\em Journal of Mathematical Economics}, 20:465--487, 1991.

\bibitem{C20}
D.~Criens.
\newblock No arbitrage in continuous financial markets.
\newblock {\em Mathematics and Financial Economics}, 14:461--506, 2020.

\bibitem{CN22}
D.~Criens and L.~Niemann.
\newblock Nonlinear continuous semimartingales.
\newblock arXiv:2204.07823v3, 2023.

\bibitem{CN22b}
D.~Criens and L.~Niemann.
\newblock Markov selections and Feller properties of nonlinear diffusions.
\newblock arXiv:2205.15200v2, 2022.

\bibitem{DS}
F.~Delbaen and W.~Schachermayer.
\newblock {\em The Mathematics of Arbitrage}.
\newblock Springer Berlin Heidelberg, 2006.

\bibitem{DM}
C.~Dellacherie and P.~A.~Meyer.
\newblock {\em Probabilities and Potential A}.
\newblock North Holland, Amsterdam, 1978.

\bibitem{denis}
L.~Denis and M.~Kervarec. 
\newblock Optimal investment under model uncertainty in nondominated models.
\newblock {\em SIAM Journal on Control and Optimization}, 51(3):1803--1822, 2013.

\bibitem{diestel}
J.~Diestel and J.~J.~Uhl, Jr.
\newblock {\em Vector Measures}.
\newblock American Mathematical Society, 1977.

\bibitem{ElKa15}
N.~El Karoui and X.~Tan.
\newblock Capacities, measurable selection and dynamic programming part II: application in stochastic control problems.
\newblock arXiv:1310.3364v2, 2015.

\bibitem{skorokhod}
I.~I.~Gikhman and A.~V.~Skorokhod.
\newblock {\em The Theory of Stochastic Processes III}.
\newblock Springer Berlin Heidelberg, reprint of the 1974 ed., 2007.

\bibitem{fadina2019affine}
T.~Fadina, A.~Neufeld, and T.~Schmidt. 
\newblock Affine processes under parameter uncertainty. 
\newblock {\em Probability, Uncertainty and Quantitative Risk}, 4(5), 2019.

\bibitem{Fed11}
S.~Federico.
\newblock A stochastic control problem with delay arising in a pension fund model.
\newblock {\em Finance and Stochastics}, 15:421–459, 2011.

\bibitem{GriMa03}
B.~Grigelionis and V.~Mackevi\u{c}ius.
\newblock The finiteness of moments of a stochastic exponential.
\newblock {\em Statistics \& Probability Letters}, 64:243--248, 2003.

\bibitem{HP91}
H.~He and N.~D.~Pearson.
\newblock Consumption and portfolio policies with incomplete markets and short-sale constraints: the infinite-dimensional case.
\newblock {\em Journal of Economic Theory}, 54(2):259--304, 1991.

\bibitem{HWY}
S.-W.~He, J.-G.~Wang and J.-A-~Yan.
\newblock {\em Semimartingale Theory and Stochastic Calculus}.
\newblock Routledge, 1992.

\bibitem{himmelberg}
C.~J.~Himmelberg.
\newblock Measurable relations.
\newblock {\em Fundamenta Mathematicae}, 87:53--72, 1975.

\bibitem{hol16}
J.~Hollender.
\newblock { \em L{\'e}vy-Type Processes under Uncertainty and Related Nonlocal Equations. }
\newblock PhD thesis, TU Dresden, 2016. 

\bibitem{IW}
N.~Ikeda and S.~Watanabe.
\newblock {\em Stochastic differential equations and diffusion processes}.
\newblock North--Holland Publishing Company Amsterdam Oxford New York, 2nd ed., 1989.

\bibitem{jacod79}
J.~Jacod.
\newblock {\em Calcul stochastique et probl\`emes de martingales}.
\newblock Springer Berlin Heidelberg New York, 1979.

\bibitem{jacod80}
J.~Jacod. 
\newblock Weak and strong solutions of stochastic differential equations.
\newblock {\em Stochastics}, 3:171--191, 1980.

\bibitem{JS}
J.~Jacod and A.~N.~Shiryaev.
\newblock {\em Limit Theorems for Stochastic Processes}.
\newblock Springer Berlin Heidelberg, 2nd ed., 2003.

\bibitem{Kallenberg}
O.~Kallenberg.
\newblock {\em Foundations of Modern Probability}.
\newblock Springer New York, 3rd ed., 2021.

\bibitem{KarKar}
I.~Karatzas and C.~Kardaras.
\newblock The num{\'e}raire portfolio in semimartingale financial models.
\newblock {\em Finance and Stochastics}, 11, 447–493, 2007.

\bibitem{KLS87}
I.~Karatzas, J.~P.~Lehoczky and S.~E.~Shreve.
\newblock Optimal portfolio and consumption decisions for a “small investor” on a finite horizon.
\newblock {\em SIAM Journal on Control and Optimization}, 25:1557--1586, 1987.

\bibitem{KLSX91}
I.~Karatzas, J.~P.~Lehoczky, S.~E.~Shreve and G.~L.~Xu.
\newblock Martingale and duality methods for utility maximisation in an incomplete market.
\newblock {\em SIAM Journal on Control and Optimization}, 29:702--730, 1991.

\bibitem{KaraShre}
I.~Karatzas and S.~E.~Shreve.
\newblock {\em Brownian Motion and Stochastic Calculus}.
\newblock Springer New York, 2nd ed., 1991.

\bibitem{kramkov}
D.~Kramkov and W.~Schachermayer
\newblock The asymptotic elasticity of utility functions and optimal investment in incomplete markets.
\newblock {\em The Annals of Applied Probability} 9(3):904--950, 1999.

\bibitem{K21}
F.~K\"uhn.
\newblock On infinitesimal generators of sublinear Markov semigroups.
\newblock {\em Osaka Journal of Mathematics}, 58(3):487--508, 2021.

\bibitem{liang}
Z.~Liang and M.~Ma.
\newblock Consumption–investment problem with pathwise ambiguity under logarithmic utility.
\newblock {\em Mathematics and Financial Economics}, 13(4):519--541, 2019.

\bibitem{liang2}
Z.~Liang and M.~Ma.
\newblock Robust consumption‐investment problem under CRRA and CARA utilities with time‐varying confidence sets.
\newblock {\em Mathematical Finance}, 30(3):1035--1072, 2020.

\bibitem{lin}
Q.~Lin and F.~Riedel.
\newblock Optimal consumption and portfolio choice with ambiguous interest rates and volatility.
\newblock {\em  Economic Theory}, 71(3):1189-1202, 2021.

\bibitem{neufeld}
C.~Liu and A.~Neufeld.
\newblock Compactness criterion for semimartingale laws and semimartingale optimal transport.
\newblock {\em Transactions of the American Mathematical Society}, 372(1):187--231, 2019.

\bibitem{GL}
G. Lowther (https://mathoverflow.net/users/1004/george-lowther).
\newblock Compactness of the set of densities of equivalent martingale measures.
\newblock {\em MathOverflow}, URL:https://mathoverflow.net/q/101784 (version: 2022-02-12).


\bibitem{memin}
J.~M{\'e}min.
\newblock Espaces de semi martingales et changement de probabilit{\'e}.
\newblock {\em Zeitschrift f{\"u}r Wahrscheinlichkeitstheorie und verwandte Gebiete}, 52(1):9--39, 1980.

\bibitem{meyer}
P.~A.~Meyer.
\newblock Limites m{\'e}diales, d'apr{\`e}s Mokobodzki.
\newblock In {\em S{\'e}minaire de Probabilit{\'e}s VII}, pages 198--204, Springer Berlin Heidelberg, 1973.

\bibitem{MU12ECP}
A.~Mijatovi\'c and M.~Urusov.
\newblock Convergence of integral functionals of one-dimensional diffusions.
\newblock {\em Electronic Communications in Probability}, 17(61):1--13, 2012.

\bibitem{neufeld2017nonlinear}
A.~Neufeld and M.~Nutz.
\newblock Nonlinear L{\'e}vy processes and their characteristics. \newblock {\em Transactions of the American Mathematical Society}, 369:69--95, 2017.

\bibitem{nutz_levy}
A.~Neufeld and M.~Nutz.
\newblock Robust utility maximization with L{\'e}vy processes. 
\newblock {\em Mathematical Finance}, 28(1):82--105, 2018.

\bibitem{normann}
D.~Normann.
\newblock Martin's axiom and medial functions.
\newblock {\em Mathematica Scandinavica}, 38(1):167--176, 1976.

\bibitem{nutz_integrals}
M.~Nutz.
\newblock Pathwise construction of stochastic integrals.
\newblock {\em Electronic Communications in Probability}, 17(24):1--7, 2012.

\bibitem{nutz}
M.~Nutz.
\newblock Random G-expectations.
\newblock {\em Annals of Applied Probability}, 23(5):1755–1777, 2013.

\bibitem{nutz_superhedging}
M.~Nutz.
\newblock Robust superhedging with jumps and diffusion.
\newblock {\em Stochastic Processes and their Applications}, 125(12):4543--4555, 2015.

\bibitem{nutz_superreplication}
M.~Nutz.
\newblock Superreplication under model uncertainty in discrete time.
\newblock {\em Finance and Stochastics}, 18(4):791--803, 2014.

\bibitem{nutz_utility}
M.~Nutz.
\newblock Utility maximization under model uncertainty in discrete time.
\newblock {\em Mathematical Finance}, 26(2):252--268, 2016. 

\bibitem{NVH}
M.~Nutz and R.~van Handel.
\newblock Constructing sublinear expectations on path space. 
\newblock {\em Stochastic Processes and their Applications}, 123(8):3100--3121, 2013.

\bibitem{park}
K.~Park and H.~Y.~ Wong.
\newblock Robust consumption-investment with return ambiguity: A dual approach with volatility ambiguity.
\newblock {\em SIAM Journal on Financial Mathematics}, 13(3):802–-843, 2022.

\bibitem{peng2010}
S.~G.~Peng.
\newblock Nonlinear expectations and stochastic calculus under uncertainty.
\newblock {\em arXiv:1002.4546}, 2010.

\bibitem{pliska86}
S.~R.~Pliska.
\newblock A stochastic calculus model of continuous trading: optimal portfolio.
\newblock {\em Mathematics of Operations Research}, 11:371--382, 1986.

\bibitem{pollard}
D.~Pollard.
\newblock {\em Convergence of Stochastic Processes}.
\newblock Springer New York, 1984.

\bibitem{RY}
D.~Revuz and M.~Yor. 
\newblock Continuous Martingales and Brownian Motion.
\newblock Springer Berlin Heidelberg, 3rd ed., 1999.

\bibitem{shiryaevProb}
 A.~N.~Shiryaev.
\newblock {\em Probability 1}.
\newblock Springer New York, 3rd ed., 2016.

\bibitem{SV}
D.~W.~Stroock and S.~R.~S.~Varadhan.
\newblock {\em Multidimensional Diffusion Processes}.
\newblock Springer Berlin Heidelberg, reprint of 1997 ed., 2006.


\end{thebibliography}
\end{document}